\documentclass[orivec,runningheads,envcountsame,envcountsect]{llncs}

\usepackage[colorlinks=true, linkcolor=blue, citecolor=blue, urlcolor=blue]{hyperref}
\usepackage{todonotes}

\usepackage[T1]{fontenc}
\usepackage[utf8]{inputenc}
\usepackage{amssymb}
\usepackage{xspace}

\usepackage[all]{xy}
\usepackage{silence}
\usepackage{subcaption}
\usepackage{bussproofs}
\usepackage{amsmath,amsfonts}
\usepackage{stmaryrd}
\SetSymbolFont{stmry}{bold}{U}{stmry}{m}{n}
\usepackage{enumitem}
\usepackage{array}
\usepackage{pifont}
\usepackage{proof}
\usepackage{fontawesome}
\usepackage{cmll}
\usepackage{braket}
\usepackage{thmtools}
\usepackage{cite}
\usepackage{bm}
\usepackage{float}
\allowdisplaybreaks

\newcommand\versalita[1]{{\normalfont\textsc{#1}}}

\newcommand\ocalc{\ensuremath{\mathcal L_{\mathcal B}}\xspace}
\newcommand\zero{\ensuremath{0}}
\newcommand\one{\ensuremath{\mathbf 1}}
\newcommand\abstr[1]{#1.}

\newcommand\inl{\mbox{\it inl}}
\newcommand\inr{\mbox{\it inr}}
\newcommand\elimone{\delta_{1}}
\newcommand\elimzero{\delta_{\zero}}
\newcommand\elimwith{\delta_{\with}}
\newcommand\elimplus{\delta_{\oplus}}

\newcommand\elimtens{\delta_{\otimes}}

\newcommand\plus{\mathbin{\text{\normalfont\scalebox{0.7}{\faPlus}}}}
\newcommand\ov[1]{\overline{#1}}
\newcommand\irule[3]{\infer[\mbox{\footnotesize $#3$}]{#2}{#1}}

\newcommand\pair[2]{\langle #1, #2 \rangle}
\newcommand\topintro{\langle\rangle}

\newcommand\pqbit{\ensuremath{\mathbf{qubit}}}

\usepackage{mathtools}
\newcommand{\defeq}{\vcentcolon=}
\newcommand\aform[1]{\mathfrak a(#1)}

\newcommand{\eval}{\ensuremath{\mathrm{eval}}}

\newcommand{\lrb}[1]{{\llbracket #1 \rrbracket}}
\newcommand\ketbra[2]{\ket{#1}\!\!\bra{#2}}
\newcommand{\pmultimap}{\boldsymbol{\multimap}}
\newcommand{\pwith}{\boldsymbol{\with}}
\newcommand\p[1]{\boldsymbol{#1}}
\newcommand{\Plus}{\mathbin{\mathord{\begin{tikzpicture}[baseline=0ex, line width=1, scale=0.18]
\draw (1,0) -- (1,2);
\draw (0,1) -- (2,1);
\end{tikzpicture}}}}

\newcommand\sdot{\p \cdot}
\newcommand\s[1]{\scalebox{#1}}  %

\newcommand\sn[1]{\textbf{SN}(#1)}

\newcommand\ul[1]{\underline{#1}}

\newcommand{\secref}[1]{Section~\ref{#1}}
\newcommand{\appref}[1]{Appendix~\ref{#1}}

\newcommand{\id}{\mathrm{id}}

\newcommand{\CC}{{\mathbf C}}

\newcommand{\FHilb}{\ensuremath{\mathbf{FHilb}}}
\newcommand{\FCstar}{\ensuremath{\mathbf{FC^*}}}

\usepackage{tikzit}

\begin{document}

\title{IMALL with a Mixed-State Modality:\texorpdfstring{\\}{} A Logical Approach to Quantum Computation}
\titlerunning{IMALL with a Mixed-State Modality}

\author{
  Kinnari Dave$^{1,2}$
  \and
  Alejandro D\'iaz-Caro$^{2,3}$
  \and
  Vladimir Zamdzhiev$^{1}$
 }

 \institute{
    Université Paris-Saclay, CNRS, ENS Paris-Saclay, Inria, Laboratoire
    Méthodes Formelles, 91190, Gif-sur-Yvette, France
  \and
    Université de Lorraine, CNRS, Inria, LORIA, F-54000 Nancy, France
  \and
    Universidad Nacional de Quilmes, Bernal, Buenos Aires, Argentina
 }

\authorrunning{K. Dave, A. D\'iaz-Caro, V. Zamdzhiev}

\maketitle

\begin{abstract}
  We introduce a proof language for Intuitionistic
  Multiplicative Additive Linear Logic (IMALL), extended with a modality
  $\mathcal B$ to capture mixed-state quantum computation. The language
  supports algebraic constructs such as linear combinations, and embeds pure
  quantum computations within a mixed-state framework via $\mathcal B$,
  interpreted categorically as a functor from a category of Hilbert Spaces to 
  a category of finite-dimensional C*-algebras.
  Measurement arises as a definable term, not as a constant, and the system
  avoids the use of quantum configurations, which are part of the
  theory of the quantum lambda calculus. Cut-elimination is defined via a composite reduction relation, and shown to be sound with respect to the denotational interpretation. We prove that any linear map on $\mathbb C^{2^n}$
  can be represented within the system, and illustrate this expressiveness with
  examples such as quantum teleportation and the quantum switch.
\end{abstract}

\section{Introduction}
\label{sec:intro}
Quantum programming languages typically adopt either \emph{classical control}
or \emph{quantum control} as a guiding principle.  Each paradigm offers
advantages but also entails limitations.  The former enables measurement and
branching over classical data, yet typically restricts quantum computation to a
fixed set of unitary operations~\cite{Selinger04,SelingerValironMSCS06}.
The latter, by contrast, supports higher-order manipulation of quantum
data and richer linear-algebraic structure, but often lacks native support for
measurement and, consequently, mixed-state
computation~\cite{DiazcaroGuillermoMiquelValironLICS19,DiazcaroMalherbeLMCS22}.

In parallel with the development of quantum programming languages, a
complementary line of research seeks to understand the logical structure
underlying quantum computation.  Just as intuitionistic logic,
$\lambda$-calculus, and Cartesian closed categories are linked via the
Curry--Howard--Lambek
correspondence~\cite{Howard1980,Lambek1980,SorensenUrzyczyn06}, one may ask
whether there exists a natural logical system whose proof language corresponds
to quantum computation.  This perspective motivates the study of \emph{quantum
computational logic}, where proofs denote quantum processes and the
type system reflects the constraints and structure of quantum theory.

From this viewpoint, the challenge is not to build yet another quantum
programming language, but to construct a logical framework with a well-defined
proof theory that aligns with the operational and denotational semantics of
quantum computation.  Such a framework should ideally support superposition and
linear combinations natively, provide a proof-theoretic foundation for unitary
transformations and entangled systems, and admit a natural account of
measurement and mixed states.

This paper presents a proof language for Intuitionistic
Multiplicative Additive Linear Logic (IMALL), extended with a modality
$\mathcal B$ that captures the transition from pure to mixed states.  Our
system is based on quantum control and supports both pure and mixed-state
computation within a unified logical setting, with a compositional and adequate
semantics.

Our approach offers three major advantages over the typical approach adopted by 
quantum lambda calculi (QLCs)~\cite{SelingerValironMSCS06, 
DaveLemonnierPechouxZamdzhievFoSSaCS25,quantum-game,enriched-presheaf}. First, 
our approach enables a richer linear-algebraic structure: instead of treating 
unitary transformations as fixed constants, they can be defined and composed as 
first-class terms, allowing precise descriptions of quantum operations as linear 
maps. Second, it allows the expression of genuinely quantum control flows, such 
as the \emph{quantum switch}~\cite{OreshkovCostaBrukner2012}, where the order of 
operations is conditioned on the state of a qubit---a behaviour not efficiently 
representable in QLCs.  Third, quantum control dispenses with the need for an 
external global configuration---such as the ``quantum configurations'' in the 
QLC---to carry quantum state alongside the program.  In our system, both state 
and program are encoded in the same proof language, and measurement is not a 
fixed constant with special semantics, but a proof term built from 
the same algebraic structure as the rest of the language.

\paragraph{Contributions.} \label{subsec:contrib}
Our main contributions are as follows:

\begin{itemize}
  \item We define the logic $\mathcal B$-IMALL by extending IMALL with a
    modality $\mathcal B$ in a \emph{transparent} way: every provable formula
    in $\mathcal B$-IMALL remains provable in IMALL when occurrences of
    $\mathcal B$ are erased.
    Our system does not include LL exponentials, because we want to keep it simple while still being able to demonstrate the main point. There can be interaction between the $\mathcal B$ functor and the LL exponentials in the model (see \cite{LZ25} which requires complicated analytic structure), but it is unclear how to do this in a simple way on the logical side. This is interesting future work.

  \item The proof language supports linear combinations of terms through the
    operators $\sdot$ and $\Plus$, which preserve the
    structure of a complex vector space
    (Corollary~\ref{cor:linearity}).
    Our system is thus \emph{algebraically linear} in a precise sense.

  \item Unlike the algebraic $\lambda$-calculi~\cite{ArrighiDowekLMCS17,
    Vaux2009, AssafEtAl2014}, our language 
    integrates both pure and mixed-state computation. The syntax accommodates
    superpositions and measurements within the same framework.

  \item Unlike the Lambda-S family of
    calculi~\cite{DiazcaroDowekRinaldiBIO19,DiazcaroMalherbeACS2020,DiazcaroGuillermoMiquelValironLICS19,
    DiazcaroMalherbeLMCS22}, which introduce a new logic based on a dedicated
    modality $S$ to manage quantum superpositions, our system builds directly
    on the well-established logic IMALL. This ensures that the logical
    foundations are familiar, modular, and reusable, while still being
    expressive enough to represent both pure and mixed-state computation.

  \item The modality $\mathcal B$ embeds the pure fragment into the mixed one,
    allowing transitions from pure-state vectors to density matrices. This
    resolves a limitation of earlier approaches using the $\odot$
    connective~\cite{DiazcaroDowekTCS23a}, which were well suited to
    probabilistic settings but lacked a clear quantum denotational
    semantics~\cite{DiazcaroMalherbe24}.

  \item The cut-elimination process is defined via a composite reduction
    relation, obtained by interleaving a 
    linear-algebraic equivalence ($\leftrightarrows$) with a deterministic
    evaluation ($\rightarrow$), as formalised in~\secref{ssec:syntax}.
    This approach combines ideas from other frameworks: the relation $\rightarrow$
    captures the interaction between algebraic constructors and logical
    connectives, as in the $\odot$ and $\mathcal L\odot$
    calculi~\cite{DiazcaroDowekMSCS24,DiazcaroDowekTCS23a}, and is here extended
    to handle the $\mathcal B$ modality. The symmetric relation
    $\leftrightarrows$ ensures canonical vector forms, in the style of algebraic
    $\lambda$-calculi~\cite{ArrighiDowekLMCS17, Vaux2009, AssafEtAl2014}.

  \item Measurement is expressible as a program, not as a constant: rather than requiring a dedicated syntactic construct with external semantics, it emerges naturally as a proof term within the language.

  \item The language avoids the need for quantum lambda configurations (QLCs)~\cite{SelingerValironMSCS06,DaveLemonnierPechouxZamdzhievFoSSaCS25,quantum-game,enriched-presheaf}, by encoding both program and quantum state uniformly as terms.

  \item We provide a denotational model inspired by~\cite{coecke2016categories}. Our model defines a functor $\mathcal B$ mapping pure-state morphisms to completely positive maps. A similar functor also
	denoted with the same symbol first appeared in~\cite{DaveLemonnierPechouxZamdzhievFoSSaCS25}. This model validates the design of the language and supports an adequacy theorem.

  \item Finally, we show that every linear transformation over $\mathbb
    C^{2^n}$ can be encoded in the pure fragment of our language
    (\secref{sssec:encoding}), and demonstrate the expressiveness of the
    language through examples such as quantum teleportation and the quantum
    switch---highlighting its ability to naturally capture paradigmatic quantum
    phenomena.
\end{itemize}

Taken together, these results yield a novel proof language
that can express both pure and mixed-state quantum computation, grounded
in a transparent logical system with a compositional and adequate semantics.

Compared to previous systems, our approach treats quantum computation not
merely as a programming discipline but as a logical theory with intrinsic
proof-theoretic content.  Rather than relying on external control structures or
dedicated constants, we encode quantum data, operations, and measurements
uniformly as logical terms.  This design bridges foundational insights from
logic with the practical demands of quantum computation.

\paragraph{Structure of the paper.}

 \secref{sec:language} introduces the syntax of the proof terms, the derivation rules for both pure and mixed fragments, and the cut-elimination procedure. We also establish standard meta-theoretic properties such as preservation under reduction, progress, and strong normalisation.
 \secref{sec:model} presents a categorical model based on the categories $\mathbf{FHilb}$ and $\mathbf{FC^*}$, which motivates the interpretation of the $\mathcal B$ modality as a transition from pure to mixed-state computation.
 \secref{sec:denotational} defines the denotational semantics of the language and proves that it is sound and adequate with respect to cut-elimination.
 \secref{sssec:encoding} shows the expressive power of the language by encoding quantum notions such as qubits, gates, measurements, higher-order quantum control, quantum teleportation, and the quantum switch.
 \secref{sec:conclusion} summarises our contributions and gives directions for future work.

\section{The \texorpdfstring{\ocalc}{L-C}-calculus}
\label{sec:language}

\subsection{Syntax, Derivation Rules, and Cut-Elimination}
\label{ssec:syntax}
The calculus $\ocalc$ is an extension of the linear lambda calculus, viewed as a proof system for an extension of IMALL.
This extension consists of adding a new modality, $\mathcal{B}$, to the language, which enables the embedding of pure-state quantum computations into the mixed-state fragment of the calculus.
As a result, the language is divided into two fragments: one for pure-state quantum computation and another for mixed-state quantum computation.
Accordingly, the logic $\mathcal{B}$-IMALL has the following grammar.
\begin{align*}
  \p P & \defeq \p \top \mid \one \mid \p P \pmultimap \p P \mid \p P \pwith \p P \mid \p{P \otimes P} & \s{0.85}{(\versalita{Pure Propositions})}\\
  A    & \defeq \zero \mid 1 \mid A\multimap A\mid A\oplus A\mid A\otimes A\mid  \mathcal B(\p P) & \s{0.85}{(\versalita{Mixed Propositions})}
\end{align*}
We denote pure propositions by bold capital Latin letters, and mixed propositions by regular capital Latin letters.

Notice that, by removing the modality $\mathcal{B}$ and ignoring the distinction between bold and non-bold connectives, the two fragments collapse into IMALL.  
For example, $\p\top \oplus \one$ is not a valid proposition in $\mathcal{B}$-IMALL, whereas $\mathcal{B}(\p\top) \oplus 1$ is.  
Removing the modality from the latter, and identifying $\one$ with $1$, yields the syntactically valid IMALL proposition $\top \oplus 1$.

The proposition $\mathcal{B}(\p P)$ represents a mixed-state quantum computation on the Hilbert space represented by the proposition $\p P$.
The notation $\mathcal{B}(\cdot)$ is inspired by our denotational model, where $\mathcal{B}(H)$ denotes the $C^*$-algebra of linear operators on a finite-dimensional Hilbert space $H$.
Readers may consult the model (Section \ref{sec:model}) to get better intuition and understanding of this.
%Since we only deal with finite-dimensional Hilbert spaces, all linear operators are bounded, and thus $\mathcal{B}(H)$ coincides with the space of bounded operators on the Hilbert space.

Note that the split between pure and mixed propositions is not justified by logical polarities or any other well-established logical criteria. Indeed, the status of logical polarities in relation to quantum theory was recently investigated in \cite{LZ25}, but only in the sense of mixed-state quantum information and entirely done in a semantic model of LL. In contrast, our work focuses on the combination of both pure and mixed quantum primitives and a large part of our work is focused on the logical treatment, not just the semantic model. 
Thus, the split between pure/mixed propositions used here is novel and inspired by the specifics of the quantum model.

The proof terms of the calculus are also divided into two fragments: one for pure-state quantum computation, corresponding to pure propositions, and another for mixed-state quantum computation, corresponding to mixed ones.
  \begin{align*}
    \p m    
    & \defeq  \p x  
    \mid \p{m \Plus m} 
    \mid \p{a} \sdot \p m 
    \mid \p{\topintro} 
    \mid \p{\ast} 
    \mid \p{\elimone( m, m)} 
    \mid \p{\lambda \abstr{x} m} 
    \mid \p{m~m} &
    \hspace{-25mm} \s{0.85}{(\versalita{Pure Terms})}\\
    &\phantom{\defeq\ } 
    \mid \p{\pair{m}{m}} 
    \mid \p{\elimwith^1(m,\abstr{x}m)} 
    \mid \p{\elimwith^2(m,\abstr{x}m)}
    \mid \p{m \otimes m} 
    \mid \p{\elimtens(m, \abstr{x y}m)} &
    \hspace{-25mm} \phantom{A}\\
    m    
    & \defeq x 
    \mid m \Plus m 
    \mid p \sdot m 
    \mid \elimzero(m)
    \mid \ast  
    \mid \elimone(m, m) 
    \mid \lambda \abstr{x}m 
    \mid m~m &\hspace{-25mm}  \s{0.85}{(\versalita{Mixed Terms})}\\
    &\phantom{\defeq\ }  
    \mid \inl(m) 
    \mid \inr(m)
    \mid \elimplus(m,\abstr{x}m,\abstr{y}m) 
    \mid m \otimes m 
    \mid \elimtens(m, \abstr{x y} m) 
    \mid \mathcal B(\p m) 
    \mid \tau(m) &  
    \hspace{-25mm} \phantom{A}
  \end{align*}
  Pure and mixed variables--denoted by $\p x, \p y, \p z$ and $x, y, z$, respectively--belong to an enumerable set $\s{0.85}{\versalita{Vars}}$.
The $\alpha$-equivalence relation and the free variables of a term are defined as usual. Terms are defined modulo $\alpha$-equivalence. 
Closed terms are the ones that do not contain free variables.  

Most terms in $\ocalc$ are standard for the linear lambda calculus. As before, to highlight the distinction between pure and mixed terms, pure terms are represented by bold letters $\p m, \p n, \p o$, while mixed terms are represented by the letters $m, n, o$.
Each construction corresponds to a rule of the logic.  

Starting with the pure fragment, the symbols $\Plus$ and $\sdot$ denote sums and scalar multiplication, which correspond to the following rules introduced in~\cite{DiazcaroDowekMSCS24}.
\[
  \irule{\p\Gamma\vdash\p P \quad \p\Gamma\vdash \p P}{\p\Gamma\vdash \p P}{}
  \qquad\qquad
  \irule{\p\Gamma\vdash\p P}{\p\Gamma\vdash \p P}{}
\]
Scalars in the pure fragment range over the field of complex numbers and are denoted by the letters $\p a, \p b, \p c$.  
We work with unnormalized vectors, since the techniques required to restrict the calculus to normalized vectors are similar to those used in~\cite{DiazcaroGuillermoMiquelValironLICS19}, and would make the calculus even more complex.

The other constructions are standard:
 The connective $\p \top$ has only an introduction rule, as usual, which is $\p \topintro$.
 The connective $\p \one$ has an introduction rule, given by $\p \ast$, and also an elimination rule, $\p{\elimone(m,n)}$, intended to be understood as sequencing. While not always necessary, it is useful in the presence of scalars.
 The connective $\p \multimap$ has the standard introduction and elimination rules: $\p{\lambda \abstr{x} m}$ and $\p{m~n}$, respectively.
 The connective $\p \with$ has an introduction rule given by $\p{\pair{m}{n}}$, and, since we present all rules in their extended version, it has the extended elimination rule $\p{\elimwith^i(m,\abstr{x}n)}$ for $i \in \{1,2\}$.
 The connective $\p \otimes$ has an introduction rule given by $\p{m \otimes n}$, and an elimination rule given by $\p{\elimtens(m,\abstr{xy}n)}$.

Similarly, the symbols $\Plus$ and $\sdot$ in the mixed fragment denote sums
and scalar multiplication. However, scalars in the mixed fragment range over
the non-negative real numbers and are denoted by the letters $p, q$. Again, we
work with unnormalized vectors, thus a generalization of probability
distributions.

The connectives $1$, $\multimap$, and $\otimes$ are completely analogous to those in the pure fragment.  
The connectives $\zero$ and $\oplus$ are standard: $\zero$ has only an elimination rule, as usual, given by $\elimzero(m)$;  
$\oplus$ has two introduction rules, $\inl(m)$ and $\inr(n)$, and an elimination rule given by $\elimplus(m,\abstr{x}n,\abstr{y}o)$.

Finally, there are two other constructions specific to the mixed fragment, which do not correspond to any logical connective:
\begin{itemize}
\item
 The modality $\mathcal B(\p m)$, which embeds pure terms into the mixed fragment. It has a rule that allows the inclusion of closed pure terms. Thus, a pure term $\p m$ can be embedded as $\mathcal B(\p m)$.
 The rule $\mathcal B^m$ behaves analogously to a modal introduction in necessity logic: from a closed derivation of $\p P$ one derives $\mathcal B(\p P)$. The other rules involving $\mathcal B$ ($\mathcal B(\pmultimap)^m_e$ and $\mathcal B(\p\otimes)^m$) are not eliminations of the modality, but rather structural rules describing how logical connectives behave under the embedding. This reflects the behaviour of the categorical functor $\mathcal B$, which maps between two distinct semantic domains: from pure-state morphisms in $\mathbf{FHilb}$ to completely positive maps in $\mathbf{FC}^*$ (see Section~\ref{sec:model}).
$\mathcal B$ can be understood as a structural operation that reflects a change of logical layer. In this sense, $\mathcal B(\p P)$ should be read as ``the mixed version of $\p P$''. Its behaviour is better captured by its categorical semantics, where $\mathcal B$ corresponds to a functor between categories representing different computational paradigms (from pure to mixed quantum computation).
\item
 The construction $\tau(m)$ is a \emph{coercion} operator that enables switching from the proposition $\mathcal B(\p P) \otimes \mathcal B(\p Q)$ to the proposition $\mathcal B(\p{P \otimes Q})$. This operator is crucial for the interaction between the pure and mixed fragments.
\end{itemize}

The pure and mixed values in this calculus are defined as follows:
\begin{align*}
  \p{v_b} & \defeq  \p{\ast} \mid \p x \mid \p{v_n \otimes v_n} \\
  \p{v_n} & \defeq  \p{\topintro} \mid \p{\lambda \abstr{x}m} \mid \p{\pair{v}{v}} \mid \p{v_b}\\ 
	\p v    & \defeq  \p{v_n} \mid  \textstyle\sum_{i=1}^n \p{a_i} \sdot \p{v_{bi}} \Plus \sum_{j=1}^l \p{v_{bj}} \mid \textstyle\sum_{i=1}^n \p 0 \sdot \p{v_{bi}} & \s{0.85}{(\versalita{Pure Values})}\\[1.2ex] 
  v_b     & \defeq  \ast \mid x \mid \inl(v_n) \mid \inr(v_n) \mid v_n \otimes v_n \mid \mathcal B(\p v)\\
  v_n     & \defeq  \lambda \abstr{x}m \mid v_b\\
	v       & \defeq  v_n \mid \textstyle\sum_{i=1}^n p_i \sdot v_{bi} \Plus \sum_{j=1}^l v_{bj}\mid \textstyle\sum_{i=1}^n 0 \sdot v_{bi} & \s{0.85}{(\versalita{Mixed Values})}
\end{align*}
where, in the summations, each $\p{v_{bi}}, \p{v_{bj}}$ (resp. $v_{bi}, v_{bj}$) is assumed to be a distinct value, and $\p{a_i} \neq \p{0},\p{1}$ (resp. $p_i > 0, p_i \neq 1$).

From now on, we will always assume that the summation symbol $\sum$ is used in the context of values, and that the summands are distinct.

Pure values $\p v$ are composed of linear combinations of the values $\p{v_b}$, and of values $\p{v_n}$, which are not linear combinations.  
The values $\p{v_n}$ correspond to the connectives $\p\top$, $\p\multimap$, and $\p\with$, for which linear combinations commute with the introduction rules and can be distributed internally. For example, the term $\pair{\p m_1}{\p n_1}\Plus\pair{\p m_2}{\p n_2}$ can be rewritten as $\pair{\p m_1\Plus\p m_2}{\p n_1\Plus\p n_2}$.  
In contrast, the values $\p{v_b}$ correspond to the connectives $\p\one$ and $\p\otimes$, where such commutation is not possible. For instance, the term $(\p m_1 \p\otimes \p n_1) \Plus (\p m_2 \p\otimes \p n_2)$ cannot be rewritten by commuting $\Plus$ with $\p\otimes$.

Similarly, mixed values $v$ are composed of linear combinations of the values $v_b$, and of values $v_n$, which are not linear combinations.

Henceforth, we may use the letters $M, N, O$ to denote \emph{both} mixed and pure terms, the letters $U, V, W$ to denote both mixed and pure values, and $S, T$ to denote both mixed and pure propositions.  
Additionally, we use the Greek letters $\alpha, \beta$ to denote scalars from either fragment.

For simplicity, we may use the same symbols for connectives in both fragments when using capital letters. For example, we may write $\otimes$ to refer to either $\p\otimes$ or $\otimes$, depending on context.  
Similarly, we may write $x$ instead of $\p x$, $\ast$ instead of $\p\ast$, and so on, when the distinction is clear.

Whenever we use metavariables like $M$, $N$, or $S$ in place of concrete terms or propositions, we assume that all components belong to the same fragment.  
For instance, since $\p{\elimone}(\p t,r)$ is not a valid term---because $\p{\elimone}$ and $\p t$ belong to the pure fragment and $r$ to the mixed one---expressions like $\elimone(M,N)$ implicitly assume that $M$, $N$, and the connective itself all belong to the same fragment.
\medskip

The logic $\mathcal{B}$-IMALL is defined by the syntax given above, together with a set of derivation rules, presented in Figure~\ref{fig:derivation_rules}.  
The rules are formulated in natural deduction style, with proof terms annotating the derivations.  

A context is a list of variable declarations, where each variable is associated with a proposition. A pure context is denoted by $\p\Gamma$ or $\p\Delta$, while a mixed context is denoted by $\Gamma$ or $\Delta$.  
As before, we simply write $\Gamma$ or $\Delta$ when the distinction is clear or when referring to both fragments.  
The notation $\Gamma, \Delta$ denotes the concatenation of the two disjoint contexts $\Gamma$ and $\Delta$.

The names of the rules include a superscript $^p$ indicating it belongs to the pure fragment or a superscript $^m$ for the mixed fragment. The absence of a superscript indicates that the rule applies to both fragments.  
In such cases, we may optionally add the superscript later when specifying one of the fragments.

\begin{figure}[t]
\centering
\[
  \irule{}
  {x : S \vdash x : S}
  {\mbox{ax}}
  \qquad
  \irule{{\Gamma \vdash M : S} & {\Gamma \vdash N : S}}
  {{\Gamma \vdash M \Plus N : S}}
  {\mbox{sum}}
  \qquad
  \irule{\Gamma \vdash M : S}
  {\Gamma \vdash \alpha \sdot M : S}
  {\mbox{prod($\alpha$)}}
\]
\[
  \irule{}
  {\p{\Gamma \vdash \topintro : \top}}
  {\mbox{$\top_i^p$}}
  \quad
  \irule{\Gamma \vdash m : \zero}
  {\Gamma , \Delta \vdash \elimzero(m) : A}
  {\mbox{$\zero_e^m$}}
  \quad
  \irule{}
  {\vdash \ast : 1}
  {\mbox{$1_i$}}
  \quad
  \irule{\Gamma \vdash M : 1 & \Delta \vdash N : S}
  {\Gamma, \Delta \vdash \elimone(M, N) : S}
  {\mbox{$1_e$}}
\]
\[
  \irule{\Gamma, x : S \vdash M : T}
  {\Gamma \vdash \lambda \abstr{x}M : S \multimap T}
  {\mbox{$\multimap_i$}}
\qquad
  \irule{\Gamma \vdash M : S \multimap T & \Delta \vdash N : S}
  {\Gamma, \Delta \vdash M~N : T}
  {\mbox{$\multimap_e$}}
\]
\[
  \irule{\p{\Gamma \vdash m : P} & \p{\Gamma \vdash n : Q}}
  {\p{\Gamma \vdash \pair{m}{n} : P \with Q}}
  {\mbox{$\pwith_i^p$}}
\]
\[
  \irule{\p{\Gamma \vdash m : P \with Q} & \p{\Delta, x : P \vdash n : R}}
  {\p{\Gamma, \Delta \vdash \elimwith^1(m,\abstr{x}n) : R}}
  {\mbox{$\pwith_{e_1}^p$}}
  \quad
  \irule{\p{\Gamma \vdash m : P \with Q} & \p{\Delta, x : Q \vdash n : R}}
  {\p{\Gamma, \Delta \vdash \elimwith^2(m,\abstr{x}n) : R}}
  {\mbox{$\pwith_{e_2}^p$}}
\]
\[
  \irule{\Delta \vdash m : A}
  {\Delta \vdash \inl(m) : A \oplus B}
  {\mbox{$\oplus_{i1}^m$}}
  \qquad
  \irule{\Delta \vdash n : B}
  {\Delta \vdash \inr(n) : A \oplus B}
  {\mbox{$\oplus_{i2}^m$}}
\]
\[
  \irule{\Delta \vdash m : A \oplus B & \Gamma , x : A \vdash n : C &  \Gamma , y : B \vdash o : C}
  {\Delta , \Gamma \vdash \elimplus(m,\abstr{x}n,\abstr{y}o) : C}
  {\mbox{$\oplus_e^m$}}
\]
\[
  \irule{\Delta \vdash M : S_1 &  \Gamma \vdash N : S_2}
  {\Delta, \Gamma \vdash M \otimes N : S_1 \otimes S_2}
  {\mbox{$\otimes_i$}}
  \qquad
  \irule{\Delta \vdash M : S_1 \otimes S_2 & \Gamma , x : S_1, y : S_2 \vdash N : T}
  {\Delta, \Gamma \vdash \elimtens(M,\abstr{x y}N) : T}
  {\mbox{$\otimes_e$}}
\]
\[
  \irule{\p{\vdash m : P}}
  {\vdash \mathcal B(\p m): \mathcal B(\p P)}
  {\mbox{$\mathcal B^m$}}
\]
\[
  \irule{\Gamma\vdash m : \mathcal B (\p P \pmultimap \p Q) & \Delta\vdash n : \mathcal B(\p P)}
  {\Gamma,\Delta\vdash m~n : \mathcal B(\p Q)}
  {\mbox{$\mathcal B(\pmultimap)_e^m$}}
  \qquad
  \irule{\Gamma\vdash m: \mathcal B(\p P) \otimes \mathcal B(\p Q)}
  {\Gamma\vdash \tau(m): \mathcal B(\p P \p \otimes \p Q)}
  {\mbox{$\mathcal B(\otimes)^m$}}
\]
\[
  \irule{\Delta , y : S_1, x : S_2 , \Gamma \vdash M : T}
  {\Delta , x : S_2, y : S_1, \Gamma \vdash M : T}
  {\mbox{ex}}
\]
\caption{Derivation rules.}
\label{fig:derivation_rules}
\end{figure}

All the rules are standard, except for the rules sum and prod$(\alpha)$, which correspond to sums and scalar multiplication, respectively, and the rules related to the $\mathcal{B}$ modality.  
Notice that the rules sum and prod$(\alpha)$ do not change the provability associated with formulae.

The rules related to $\mathcal{B}$ are the following:
\begin{itemize}
  \item The rule $\mathcal{B}^m$ allows the introduction of \emph{closed} pure terms into the mixed fragment. There are two reasons for this being restricted to closed terms.  
    First, the design and construction of proof terms for propositions $\mathcal{B}(\p P)$ are inspired by the denotational model, where $\mathcal{B}$ is interpreted as a functor between two categories. As such, it cannot be composed directly with morphisms from the domain category.
    Second, allowing open proof terms of propositions $\p P$ to be promoted to proof terms of $\mathcal{B}(\p P)$ would permit arbitrary transformations from pure to mixed propositions to be expressed as valid terms in $\ocalc$.
    This goes against the intended design, which views the pure fragment as a subsystem of $\ocalc$ dedicated to pure-state quantum computation.  
    In contrast, the mixed fragment is regarded as the ``main'' language of the calculus, with the pure fragment \emph{included} into it to provide expressive power for encoding pure-state quantum computation.

  \item The rule $\mathcal{B}(\pmultimap)_e^m$ allows the elimination of an embedded pure arrow.  

  \item The rule $\mathcal{B}(\otimes)^m$ is a coercion rule. It allows one to transition from the proposition $\mathcal{B}(\p P) \otimes \mathcal{B}(\p Q)$ to the less informative proposition $\mathcal{B}(\p{P \otimes Q})$---just as type-casting an integer to a real ``forgets'' its integer nature.  
\end{itemize}

Cut-elimination in $\ocalc$ is defined via a reduction relation $\hookrightarrow$ (Definition~\ref{def:hook}), formulated in terms of two auxiliary relations: an equivalence relation $\leftrightarrows$, capturing basic algebraic identities between proof terms, and a cut-reduction relation $\rightarrow$. Contexts are defined by the following grammar:
  \begin{align*}
    K & \defeq  [.] 
    \mid K \Plus M 
    \mid \alpha \sdot K 
    \mid \elimzero(K)
    \mid \elimone(K,M) 
    \mid K~M 
    \mid V~K \\ 
    &\phantom{\defeq\ } 
    \mid \p \langle K\p{,m\rangle} 
    \mid \p{\langle v,}K \p \rangle 
    \mid \p{\elimwith^1(}K \p{, \abstr{x}m)}
    \mid \p{\elimwith^2(}K\p{, \abstr{x}t)} \\
    &\phantom{\defeq\ }
    \mid \inl(K) 
    \mid \inr(K) 
    \mid \elimplus(K,x.m, y.m)\\
    &\phantom{\defeq\ }
    \mid K \otimes M 
    \mid V \otimes K 
    \mid \elimtens(K, \abstr{xy}M) 
    \mid \mathcal B(K)
    \mid \tau(K) 
  \end{align*}
 
  The equivalence relation $\leftrightarrows$ on terms is defined in Figure~\ref{fig:bidir_rules} (top part).  The rules
are divided into two groups: the left column corresponds to the rules usually
seen in algebraic lambda
calculi~\cite{AssafEtAl2014,ArrighiDowekLMCS17,Vaux2009,DiazcaroGuillermoMiquelValironLICS19},
while the right column corresponds to rules involving the $\otimes$ connective.

\begin{figure}[!ht]
  \centering
    \parbox{0.55\textwidth}{
% \begingroup\setlength{\jot}{0.05em}
      \begin{align*}
	1 \sdot M &\leftrightarrows M\\
	M_1 \Plus M_2 &\leftrightarrows M_2 \Plus M_1\\
	M_1 \Plus (M_2 \Plus M_3) &\leftrightarrows (M_1 \Plus M_2) \Plus M_3\\
	\alpha \sdot M \Plus \beta \sdot M &\leftrightarrows (\alpha + \beta) \sdot M\\
	\alpha \sdot (\beta \sdot M) &\leftrightarrows \alpha\beta \sdot M\\
	\alpha \sdot (M_1 \Plus M_2) &\leftrightarrows \alpha \sdot M_1 \Plus \alpha \sdot M_2\\
	\text{If }M \neq 0 \sdot M',\ M \Plus 0 \sdot N &\leftrightarrows M
      \end{align*}
 %\endgroup
    }
    \parbox{0.44\textwidth}{
      \begin{align*}
	M_1 \otimes (\alpha\sdot M_2) &\leftrightarrows \alpha \sdot M_1 \otimes M_2\\
	(\alpha \sdot M_1) \otimes M_2 &\leftrightarrows \alpha \sdot M_1 \otimes M_2 \\
	(M_1 \Plus M_2) \otimes N &\leftrightarrows M_1 \otimes N \Plus M_2 \otimes N\\
	M \otimes (N_1 \Plus N_2) &\leftrightarrows M \otimes N_1 \Plus M \otimes N_2
      \end{align*}
    }
%\begingroup\setlength{\jot}{0.05em}
    \begin{align} 
      \elimone(\alpha\sdot \ast,M) &\rightarrow \alpha \sdot M\label{rule:one} \\
      (\lambda \abstr{x}M)~V &\rightarrow (V/x)M\label{rule:arrow} \\
      \p{\elimwith^1(\pair{v}{w}, {\abstr{x}t})} &\rightarrow \p{( v/ x) t}\label{rule:withleft}\\
      \p{\elimwith^2(\pair{v}{w}, {\abstr{y}t})} &\rightarrow \p{( w/ y) t}\label{rule:withright}\\
      \elimplus(\inl(v),\abstr{x}m,\abstr{y}n) & \rightarrow  (v/x)m\label{rule:oplusleft}\\
      \elimplus(\inr(v),\abstr{x}m,\abstr{y}n) & \rightarrow  (v/y)n\label{rule:oplusright}\\
      \elimtens(V_{n} \otimes W_{n},\abstr{x y}M) & \rightarrow (V_{n}/x,W_{n}/y)M \label{rule:tensor}
      \\[1.2ex]
      \p a \sdot \p \topintro &\rightarrow \p \topintro \label{rule:commtop}\\
      \textstyle\sum_{i=1}^n \alpha_i \sdot \lambda \abstr{x}M_i &\rightarrow \lambda \abstr{x}(\textstyle\sum_{i=1}^n\alpha_i \sdot M_i) \label{rule:commarrow}\\
      \textstyle\sum_{i=1}^n \p{a_i} \sdot \p{\pair{v_{i}}{w_{i}}} &\rightarrow \p{\langle}\textstyle\sum_{i=1}^n \p{a_i} \sdot \p{v_{i},}\textstyle\sum_{i=1}^n \p{a_i} \sdot \p{w_{i}\rangle} \label{rule:commwith}\\
      \elimplus(\textstyle\sum_{i=1}^n p_i \sdot v_{bi},\abstr{x}o,\abstr{y}s) &\rightarrow \textstyle\sum_{i=1}^n p_i \sdot \elimplus(v_{bi},\abstr{x}o,\abstr{y}s) \label{rule:commoplus}\\
      \elimtens(\textstyle\sum_{i=1}^n\alpha_i \sdot V_{bi},\abstr{x y}M) &\rightarrow \textstyle\sum_{i=1}^n \alpha_i \sdot \elimtens(V_{bi},\abstr{x y}M) \label{rule:commtensor}\\
      \inl(\textstyle\sum_{i=1}^n p_i \sdot v_{bi}) &\rightarrow \textstyle\sum_{i=1}^n p_i \sdot \inl(v_{bi}) \label{rule:comminl}\\
      \inr(\textstyle\sum_{i=1}^n p_i \sdot v_{bi}) &\rightarrow \textstyle\sum_{i=1}^n p_i \sdot \inr(v_{bi}) \label{rule:comminr}
      \\[1.2ex]
      \tau(\mathcal B(\p{v}) \otimes \mathcal B(\p{w})) &\rightarrow \mathcal B(\p{v \otimes w})\label{rule:casting}\\
      \tau(\textstyle\sum_{i=1}^n p_i \sdot v_{bi}) &\rightarrow \textstyle\sum_{i=1}^n p_i \sdot \tau(v_{bi}) \label{rule:commtau}\\
      \text{If }\p a \neq \p 1,\ \mathcal B(\p a \sdot \p{v_b}) &\rightarrow |\p a|^2 \sdot \mathcal B(\p{v_b})\label{rule:commB}\\ 
      (\textstyle\sum_{i=1}^np_i \sdot \mathcal B(\p{v_i}))~(\textstyle\sum_{j=1}^kq_j \sdot \mathcal B(\p{w_j})) &\rightarrow \textstyle\sum_{i,j}p_iq_j \sdot \mathcal B(\p{v_i~w_j}) \label{rule:commBapp}
    \end{align}
%\endgroup
    \[\kern3.6em
      \infer[]
      {K[M] \leftrightarrows K[N]}
      {M \leftrightarrows N}
      \qquad
      \infer[]
      {K[M] \rightarrow K[N]}
      {M  \rightarrow N}
    \]
  \caption{Reduction relations $\leftrightarrows$ and $\rightarrow$.}
  \label{fig:red_rightarrow}
  \label{fig:bidir_rules}
\end{figure}

\begin{definition}[$\hookrightarrow$]
  \label{def:hook}
  The relation $\hookrightarrow$ %$\cdot \hookrightarrow \cdot \subseteq (\s{0.85}{\versalita{Pure Terms}} \times \s{0.85}{\versalita{Pure Terms}}) \cup (\s{0.85}{\versalita{Mixed Terms}} \times \s{0.85}{\versalita{Mixed Terms}})$ 
  is defined as:
  \begin{center}
    $\vcenter{
    \infer[]
    {M \hookrightarrow N}
    {M \rightarrow N}
  }
    \qquad
    \vcenter{
    \infer[]
    {M \hookrightarrow V}
    {M \leftrightarrows V}
  }
    \qquad
  \vcenter{
    \infer[]
    {M \hookrightarrow N} 
    {M \not\rightarrow \quad M\leftrightarrows \aform{M} \quad \aform{M} \rightarrow N}
  }$
  \end{center}

  We denote by $\equiv$ the reflexive, transitive, and symmetric closure of $\hookrightarrow$.
\end{definition}

%\begin{figure}[t]         
%\centering
%\parbox{0.55\textwidth}{
%% \begingroup\setlength{\jot}{0.05em}
%  \begin{align*}
%    1 \sdot M &\leftrightarrows M\\
%    M_1 \Plus M_2 &\leftrightarrows M_2 \Plus M_1\\
%    M_1 \Plus (M_2 \Plus M_3) &\leftrightarrows (M_1 \Plus M_2) \Plus M_3\\
%    \alpha \sdot M \Plus \beta \sdot M &\leftrightarrows (\alpha + \beta) \sdot M\\
%    \alpha \sdot (\beta \sdot M) &\leftrightarrows \alpha\beta \sdot M\\
%    \alpha \sdot (M_1 \Plus M_2) &\leftrightarrows \alpha \sdot M_1 \Plus \alpha \sdot M_2\\
%    \text{If }M \neq 0 \sdot M',\ M \Plus 0 \sdot N &\leftrightarrows M
%  \end{align*}
% %\endgroup
%  }
%\parbox{0.44\textwidth}{
%  \begin{align*}
%    M_1 \otimes (\alpha\sdot M_2) &\leftrightarrows \alpha \sdot M_1 \otimes M_2\\
%    (\alpha \sdot M_1) \otimes M_2 &\leftrightarrows \alpha \sdot M_1 \otimes M_2 \\
%    (M_1 \Plus M_2) \otimes N &\leftrightarrows M_1 \otimes N \Plus M_2 \otimes N\\
%    M \otimes (N_1 \Plus N_2) &\leftrightarrows M \otimes N_1 \Plus M \otimes N_2
%  \end{align*}
%}
%\infer[]
%       {K[M] \leftrightarrows K[N]}
%       {M \leftrightarrows N}
%\caption{Equivalence relation $\leftrightarrows$.}
%\label{fig:bidir_rules}
%\end{figure}

The relation $\leftrightarrows$ gives rise to the possibility of defining a canonical form for terms.  
To this end, we introduce the notion of \emph{algebraic form} of a term, which is a linear combination of distinct \emph{base forms} of terms.  
These concepts are defined by the following grammar. As before, the symbols correspond to both fragments unless otherwise specified.  
We recall that, as per our earlier convention, all summations are assumed to involve pairwise distinct summands.
\begin{align*}
  \mathfrak b &
  \defeq x 
  \mid \p{\topintro} 
  \mid \elimzero(\mathfrak a) 
  \mid \ast 
  \mid \elimone(\mathfrak a,\mathfrak a) 
  \mid \lambda \abstr{x} \mathfrak{a} 
  \mid \mathfrak{a}~\mathfrak{a} 
  \mid \p{\langle}\mathfrak{a}\p ,\mathfrak{a}\p{\rangle} 
  \mid \p{\elimwith^i(}\mathfrak{a}\p{, \abstr{x}}\mathfrak{a}\p)
  &\hspace{-12mm} \s{0.85}{(\versalita{Base Form})}
  \\
  &\phantom{\defeq}
  \mid \inl(\mathfrak{a}) 
  \mid \inr(\mathfrak{a}) 
  \mid \elimplus(\mathfrak{a}, \abstr{x}\mathfrak{a}, \abstr{y}\mathfrak{a}) 
  \mid \mathfrak{b} \otimes \mathfrak{b} 
  \mid \elimtens(\mathfrak{a}, \abstr{xy}\mathfrak{a}) 
  \mid \mathcal B(\mathfrak{a}) 
  \mid \tau(\mathfrak{a})
  \\
  \mathfrak a &
  \defeq \textstyle\sum_{i=1}^n 0 \sdot \mathfrak{b}_i 
  \mid \textstyle\sum_{i=1}^n \alpha_i \sdot \mathfrak{b}_i~(\alpha_i \neq 0)
  &\hspace{-12mm}\s{0.85}{(\versalita{Algebraic Form})}
\end{align*}
Note that the choice of a symmetric relation and the use of algebraic forms are made purely for convenience.  
The rules could be oriented, following~\cite{AssafEtAl2014}, so that canonical forms coincide with normal forms.
\begin{restatable}[Uniqueness]{theorem}{Uniqueness}
  \label{lem:unique_alg}
  If $\Gamma \vdash M : T$, then there is a unique algebraic form $\mathfrak a$
  such that $M \leftrightarrows \mathfrak a$.  We denote the algebraic form of
  $M$ by $\aform{M}$.
\end{restatable}
\begin{proof}
  Proof is by induction on $M$. See \appref{app:proofs} for details.
\end{proof}
Note that the values given in \secref{ssec:syntax} are in their algebraic forms.

The rewrite relation $\rightarrow$ is also defined in Figure~\ref{fig:red_rightarrow}.  
As usual, we write $(N/x)M$ for the substitution of $N$ for $x$ in $M$.  
In the rules involving summations, we assume---as per our earlier convention---that all summands are pairwise distinct.  
The rules are divided into three groups:
\begin{itemize}
  \item Rules~\eqref{rule:one} to~\eqref{rule:tensor}
    correspond to cut-elimination steps: an
    introduction rule followed by an elimination rule for each connective.  The
    only non-standard case is the rule for $1$ (and $\one$), which carries a
    scalar $\alpha$ in the elimination.

  \item Rules~\eqref{rule:commtop}
    to~\eqref{rule:commtensor},  consist of
    commutation rules that allow the proof constructors $\Plus$ and $\sdot$ to commute
    with each connective.  In some cases, commutation applies to introduction
    rules (for $\top$, $\multimap$, and $\with$); in others, to elimination
    rules (for $\oplus$ and $\otimes$).  There is no commutation rule for
    $\zero$, and no need for one for $1$, which is already handled by
    Rule~\eqref{rule:one}.  The last two rules of this group,
    Rules~\eqref{rule:comminl}
    and~\eqref{rule:comminr}, are not strictly
    necessary, as they are subsumed by
    Rule~\eqref{rule:commoplus}.  However, they are useful for
    obtaining better normal forms.  For instance, the term $\inl(p \sdot \ast)
    \Plus \inl(q \sdot \ast)$ reduces, using $\rightarrow$ and
    $\leftrightarrows$, to $(p + q) \sdot \inl(\ast)$.

  \item Rule~\eqref{rule:casting}
    is a coercion rule to transition from
    proposition $\mathcal{B}(\p P) \otimes \mathcal{B}(\p Q)$  to 
    proposition $\mathcal{B}(\p{P \otimes Q})$.  Rule~\eqref{rule:commtau}
    is a commutation rule for the coercion operator $\tau$.
    Rule~\eqref{rule:commB} allows a scalar to commute with the
    $\mathcal{B}$ operator.  Note that $\mathcal{B}$ is not a linear operator.
    Intuitively, if $\p{v_b}$ represents a quantum state in vector form, then
    $\mathcal{B}(\p{v_b})$ represents the same state in its density matrix
    form.  As a consequence, the scalar $\p a$ is multiplied by $\overline{\p a}$ when commuting with
    $\mathcal{B}$.  Finally, Rule~\eqref{rule:commBapp} distributes
    applications over the $\mathcal{B}$ operator.
\end{itemize}

Intuitively, $\hookrightarrow$ interleaves $\leftrightarrows$ and
$\rightarrow$: it rewrites any term via $\leftrightarrows$ until it either
reduces under $\rightarrow$ or reaches normal form. For instance, the term
$\elimplus(\inl(\ast) \Plus \inl(\ast), \abstr{x}m, \abstr{y}n)$ does not
reduce under $\rightarrow$ alone (because $\inl(\ast)\Plus\inl(\ast)$ is not a value), but rewriting it via $\leftrightarrows$ to
$\elimplus(2 \sdot \inl(\ast), \abstr{x}m, \abstr{y}n)$ makes it reducible.
Each step of $\leftrightarrows$ is a purely algebraic manipulation. In the
following section, we show that $\hookrightarrow$ satisfies the standard
properties of a well-behaved reduction system.

\subsection{Correctness}
We state the standard correctness theorems. See \appref{app:proofs} for omitted proofs.

\begin{restatable}[Progress]{theorem}{progress}
  \label{thm:prog}
  $\vdash M : T$ implies $T \neq \zero$, and either $M$ is a value $V$, or there exists $N$ such that $M \hookrightarrow N$.
  \qed
\end{restatable}

\begin{restatable}[Preservation]{theorem}{subred}
  \label{thm:sub-red}
  $\Gamma \vdash M : T$ and $M \hookrightarrow N$ imply $\Gamma \vdash N : T$.
  \qed
\end{restatable}

\begin{restatable}[Strong Normalisation]{theorem}{strongnorm}
  \label{cor:strong-norm}
  $\Gamma \vdash M : T$ implies $M$ terminates.
\end{restatable}
\begin{proof}
  The proof uses Girard's ultra-reduction technique~\cite{Girard72}, extending
  the relation with $\Plus$-elimination and scalar erasure rules. Strong
  normalisation of the extended system implies that of the original. See
  \appref{app:SN}. \qed
\end{proof}

\begin{theorem}[Confluence]
  \label{thm-confluence}
  Let $\Gamma \vdash M : T$.  
  If $M \hookrightarrow^* M_1$ and $M \hookrightarrow^* M_2$,  
  then there exists a value $V$ such that $M_1 \hookrightarrow^* V \ ^*\!\!\hookleftarrow M_2$.
\end{theorem}
\begin{proof}
  The confluence of $\hookrightarrow$ follows from its structure as an interleaving of a deterministic reduction ($\rightarrow$) and the equivalence $\leftrightarrows$. 
The relation $\leftrightarrows$ is applied only when $\rightarrow$ is blocked, serving as an algebraic normalisation step to enable further reduction.
By Theorems~\ref{lem:unique_alg} and~\ref{thm:prog}, the normal forms of $\hookrightarrow$ are unique, so any two reduction sequences from $M$ must converge.
  \qed
\end{proof}

\section{Categorical Model}
\label{sec:model}

This section introduces the categorical model used for our denotational
semantics, structured in two parts. The first, described in
Subsection~\ref{sub:FHilb},
consists of finite-dimensional Hilbert spaces and linear maps---a standard
setting for pure-state quantum computation---and is used to interpret the pure
fragment of $\ocalc$. The second, presented in
Subsection~\ref{sub:FCstar},
consists of finite-dimensional C*-algebras and completely positive maps, which
model the mixed fragment. We describe the relevant categorical structure in 
Subsection \ref{sub:categorical}.

\subsection{Finite-dimensional Hilbert spaces}
\label{sub:FHilb}

Finite-dimensional Hilbert spaces are ubiquitous and well-known in quantum
information theory, so we mostly use this subsection to fix notation.

\begin{definition}
  A \emph{finite-dimensional Hilbert space} is a vector space $H$ equipped with
  an inner-product $\langle \cdot | \cdot \rangle \colon H \times H \to \mathbb
  C$ that is anti-linear in the first argument and linear in the second one.
\end{definition}

If $H_1$ and $H_2$ are two finite-dimensional Hilbert spaces, we write $H_1
\otimes H_2$ for their \emph{Hilbert space tensor product} and we write $H_1
\oplus H_2$ for their \emph{Hilbert space direct sum}, both defined in the
standard way.
The maps that are relevant for our development are the linear maps $f \colon
H_1 \to H_2$ between finite-dimensional Hilbert spaces. Since $H_1$ and $H_2$
are finite-dimen\-sional, every such map $f$ is necessarily bounded, or
equivalently, continuous with respect to the usual topology.
The \emph{adjoint} of $f$ is the (bounded) linear map $f^\dagger \colon H_2 \to H_1$ that is uniquely determined by the property:
\( \langle f(h_1) | h_2 \rangle = \langle h_1 | f^\dagger(h_2) \rangle \) for every  $h_1 \in H_1$ and $h_2 \in H_2$.
If we choose orthonormal bases of $H_1$ and $H_2$, then the matrix
representation of $f^\dagger$ is given by taking the conjugate transpose of the
matrix representation of $f.$ We say that a linear map $f \colon H_1 \to H_2$ is \emph{unitary} if $f \circ f^\dagger = \id_{H_2}$ and $f^\dagger \circ f = \id_{H_1}.$

\subsection{Finite-dimensional C*-algebras}
\label{sub:FCstar}

Finite-dimensional C*-algebras \cite{takesaki} have sufficient structure to allow us to
describe completely-positive maps and to model the quantum operations (also
known as channels) that are used for mixed-state quantum computation in finite
dimensions. We begin by recalling the relevant definitions.

\begin{definition}
  A finite-dimensional $C^*$-algebra is a complex vector space $A$ together with:
  a binary operation $( - \cdot - ) : A \times A \rightarrow A$, called multiplication, and written via juxtaposition, i.e. $ab \defeq a \cdot b$, which is associative and linear in both components;
  a multiplicative unit $1 \in A$, such that $1 a = a 1 = a$ for all $a \in A$;
  a unary operation $(\cdot)^* \colon A \to A$, called involution, such that $(a^*)^* = a$, $(ab)^* = b^*a^*$, $(\lambda a)^* = \bar{\lambda}a^*$,
  and $(a + b)^* = a^* + b^*$, for all $a, b \in A$ and $\lambda \in \mathbb{C}$;
  a norm $\|\cdot\| \colon A \to [0,\infty)$ that satisfies certain additional conditions.
\end{definition}

\begin{remark}
  We do not make essential use of C*-algebra norms in this paper, so we omit the details of this part of the definition.
  In order to avoid repetition, we often say ``C*-algebra'' instead of ``finite-dimensional C*-algebra''.
\end{remark}

Let $A$ be a C*-algebra. An element $a \in A$ is called \emph{self-adjoint} if
$a =a^*.$ We say that $a \in A$ is \emph{positive} if there exists an element
$b \in A$, such that $a = bb^*.$ Note that every positive element is
self-adjoint.

\begin{example}
  One important example of a C*-algebra is given by the matrix algebra $M_n$,
  consisting of the $n \times n$ complex matrices. The vector space structure
  is obvious and multiplication is given by multiplication of matrices. The
  involution is given by the conjugate transpose and the unit element $1$ is the
  identity matrix $1_n$. The norm on $M_n$ is the usual operator norm. The
  complex numbers $\mathbb C$ is a C*-algebra which may be identified with
  $M_1$.
\end{example}

\begin{example}
  \label{ex:B(H)}
  If $H$ is a Hilbert space, then the space
  $\mathcal B(H)$ of bounded linear operators on $H$,
  has the structure of a C*-algebra. The vector space structure is obvious,
  the involution is given by taking adjoints, the multiplication by composition
  of functions, the multiplicative unit is $\id_H$ and the norm is the usual
  operator norm. When $H$ is finite-dimensional, as assumed in this paper, with dimension $n$, then $\mathcal B(H) \cong M_n$ as C*-algebras, so the
  two may be often identified.
\end{example}

\begin{definition}
  Let $A$ and $B$ be C*-algebras. The direct sum of $A$ and $B$ is the C*-algebra $A \oplus B$ defined in the following way:
  the vector space structure is the vector space direct sum $A \oplus B$;
  the involution, multiplication and unit are defined pointwise;
  the norm is defined as $\| (a,b) \| \defeq \mathrm{max}(\| a \|, \| b \|).$
\end{definition}

\begin{definition}
\label{def:tens_alg}
  Given two finite-dimensional C*-algebras $A$ and $B$, we write $A \otimes B$ for their C*-algebra tensor product which is defined
  in the following way:
  the underlying vector space is $A \otimes B$;
  the multiplication map is determined by the assignment $(a_1 \otimes b_1)(a_2 \otimes b_2) = a_1a_2 \otimes b_1b_2$;
  the multiplicative unit is $1_A \otimes 1_B$;
  the involution is determined by the assignment $(a \otimes b)^* = a^* \otimes b^*$;
  there is a unique way to assign a C*-algebra norm \cite[pp. 231]{pisier}, but we elide the details.
\end{definition}

We can now define the C*-algebra morphisms that we use in our semantics.

\begin{definition}
  Given C*-algebras $A$ and $B$, we say that a linear map $f \colon A \to B$ is
  \emph{positive} when $f$ preserves positive elements and
  \emph{completely-positive} when the map $\id_{M_n} \otimes f \colon M_n \otimes
  A \to M_n \otimes B$ is positive for every $n \in \mathbb N.$
\end{definition}

\begin{example}
  \label{ex:cp}
  For every C*-algebra $A$, the identity map $\id_A$ is completely-pos\-itive. If $H, K$ are finite-dimensional Hilbert spaces and
  $f \colon H \to K$ is a linear map, then the map $f(\cdot)f^\dagger \colon B(H) \to B(K)$ is a completely-positive map.
\end{example}

\subsection{Categorical Structure}
\label{sub:categorical}

We write $\FHilb$ for the category of finite-dimensional Hilbert spaces with
linear maps as morphisms and we write $\FCstar$ for the category
of finite-dimensional C*-algebras and completely-positive maps as
morphisms. Both categories enjoy similar categorical properties that we now
recall and for which we use similar notation. See \cite{heunen2019categories,coecke2016categories} for more information.
Throughout the remainder of the subsection, let $X,Y,Z$ be finite-dimensional
Hilbert spaces (C*-algebras).

\emph{Symmetric Monoidal Structure.} The category $\FHilb$ $(\FCstar)$ has a symmetric
monoidal structure with tensor unit given by the Hilbert space (C*-algebra)
$\mathbb C$ and monoidal product given by the Hilbert space (C*-algebra) tensor
product $X \otimes Y.$ The left (right) unitors $\lambda$ $(\rho)$, the
associator $\alpha$, and symmetry $\sigma$ natural isomorphisms
are defined in the same way as for vector spaces.

\emph{Compact Closed Structure.}
The category $\FHilb$ $(\FCstar)$ is also compact closed. We write $\eta_X
\colon \mathbb C \to X^* \otimes X$ for the unit and $\epsilon_X \colon X
\otimes X^* \to \mathbb C$ for the counit of the compact closed structure,
where $X^*$ indicates the dual of $X$ (in the sense of compact closure). It
follows that $\FHilb$ $(\FCstar)$ is a closed symmetric monoidal category with
internal hom given by $[X, Y] \defeq X^* \otimes Y$. We write
$\Phi \colon \CC(X \otimes Y, Z) \cong \CC(X, [Y,Z])$ for the currying natural
isomorphism, where $\CC$ stands for $\FHilb$ or $\FCstar,$ and we write
$\eval_{X,Y} \colon [X,Y] \otimes X \to Y$ for the canonically induced
evaluation morphism of $\CC.$

\emph{Finite Biproducts.}
The category $\FHilb$ $(\FCstar)$ also has finite biproducts. Binary biproducts
are given by the Hilbert space (C*-algebra) direct sum $X \oplus Y$ and the
zero object is given by the zero-dimensional Hilbert space (C*-algebra) 0.
We write $!_{X,Y} \colon X \to Y$ for the unique morphism that factors through $0$.
We write $\pi_1 \colon X \oplus Y \to X$ 
and $\pi_2 \colon X \oplus Y \to Y$ for the canonical projections, which are defined as
$\pi_1(x,y) = x$ and $\pi_2(x,y) = y$. Given morphisms $f \colon Z \to X$ and
$g \colon Z \to Y$ in $\FHilb$ $(\FCstar)$, $\langle f,g \rangle
\colon Z \to X \oplus Y$ is the canonical morphism induced by the categorical
product and defined by $\langle f,g \rangle(z) \defeq (f(z), g(z))$. We write
$i_1 \colon X \to X \oplus Y$ and $i_2 \colon Y \to X \oplus Y$ for the
canonical coproduct injections defined by $i_1(x) \defeq (x, 0)$ and $i_2(y) =
(0,y).$ Given morphisms $f \colon X \to Z$ and $g \colon Y \to Z$ in $\FHilb$
$(\FCstar)$, we write $[f,g] \colon X \oplus Y \to Z$ for the canonical map
induced by the couniversal property of $X \oplus Y$ and defined by $[f,g](x,y)
= f(x) + g(y).$ Since $\FHilb$ $(\FCstar)$ is symmetric monoidal closed, it follows
that the monoidal product distributes over coproducts and we write
$d_{X,Y, Z} \colon X \otimes (Y \oplus Z) \cong (X \otimes Y) \oplus (X \otimes Z)$
for the canonical natural isomorphism.

\emph{Linear combinations.}
The homsets $\FHilb(X, Y)$ are closed under
finite $\mathbb C$-linear combinations, i.e. for $f_i \in \FHilb(X,Y)$,
$\alpha_i \in \mathbb C$ complex scalars, $i \in I$ with
$I$ a finite set, $\sum_{i \in I} \alpha_i f_i \in \FHilb(H,K)$, where the
sum is defined pointwise in the usual way.
However, the homsets $\FCstar(X,Y)$ are \emph{not} closed under such linear combinations,
because if $f \colon X \to Y$ is a completely-positive map, the map $(-1)f$ need not be one.
Instead, homsets in $\FCstar(X,Y)$ are closed under finite $\mathbb R_{\geq 0}$-linear combinations,
again defined pointwise, so that $\sum_{i \in I} r_i f_i \in \FCstar(X,Y)$,
where $r_i \geq 0$ for all $i \in I.$
Linear combinations behave well with these categorical constructions, as the next lemma shows.
\begin{lemma}
\label{lem:pm_eq}
  Given compatible morphisms $f,f' ,g, g', h$ in $\FHilb$ $(\FCstar)$ and a scalar $a \in \mathbb C$ $(a \in \mathbb R_{\geq 0}),$ the following equations hold:
  \[
    \begin{array}{rl@{\qquad}rl}
      \Phi (f \Plus g)                                  &= \Phi(f) \Plus \Phi(g)                       & \Phi (a \sdot f)             &= a \sdot \Phi(f) \\
      \langle f, g\rangle \Plus \langle f' , g' \rangle &= \langle f \Plus f' , g \Plus g' \rangle     & a \sdot \langle f ,g \rangle &= \langle a \sdot f, a \sdot g \rangle \\
      {[} f, g {]} \Plus {[} f' , g' {]}                        &= {[} f \Plus f' , g \Plus g' {]}                 & a \sdot {[} f ,g {]}             &= {[} a \sdot f, a \sdot g {]}  \\
      (f \Plus g) \otimes h                             &= f \otimes h \Plus g \otimes h               &a \sdot (f \otimes h)                             &= (a \sdot f) \otimes h = f \otimes (a \sdot h) \\
      h \otimes (f \Plus g)        &= h \otimes f \Plus h \otimes g  
    \end{array}
  \]
\end{lemma}
\begin{proof}
Straightforward verification.
\qed
\end{proof}

The assignment $\mathcal B(-)$ from Example \ref{ex:B(H)} can be extended
functorially by defining $\mathcal B(f) \defeq f(\cdot)f^\dagger$ for a linear
map $f \colon H \to K$, see Example \ref{ex:cp}. This allows us to lift
functions acting on pure states to ones acting on mixed states.

\begin{lemma}
  The assignment $\mathcal B(-) \colon \FHilb \to \FCstar$ can be equipped with
  the structure of a strong monoidal functor. We write $\tau_{H,K} \colon \mathcal B(H)
  \otimes \mathcal B(K) \cong \mathcal B(H \otimes K)$ and $\tau_{\mathbb C} \colon \mathbb C \cong
  \mathcal B(\mathbb C)$ for the obvious natural isomorphisms.
  \qed
\end{lemma}

\section{Denotational semantics}
\label{sec:denotational}

With the model in place, we now describe the denotational semantics of $\ocalc$. Pure and mixed propositions are interpreted as objects in $\textbf{FHilb}$ and $\textbf{FC*}$, respectively.
\begin{gather*}
\lrb{\p \top} \defeq 0_{\textbf{FHilb}}, \hspace{3.4mm} \lrb{\one} \defeq \mathbb{C}, \hspace{3.4mm} \lrb{\p P \p \multimap \p Q} \defeq [\lrb{\p P},\lrb{\p Q}], \hspace{3.4mm} 
	\lrb{\p P \pwith \p Q} \defeq \lrb{\p P} \oplus \lrb{\p Q},  \notag \\
\lrb{\p P \p \otimes \p Q} \defeq \lrb{\p P} \otimes \lrb{\p Q}, \quad \lrb{\zero} \defeq 0_{\textbf{FC*}}, \quad \lrb{1} \defeq \mathbb{C}, 
	\quad \lrb{A \multimap B} \defeq [\lrb{A},\lrb{B}], \notag \\
\lrb{A \oplus B} \defeq \lrb{A} \oplus \lrb{B}, \quad \lrb{A \otimes B} \defeq \lrb{A} \otimes \lrb{B}, \quad \lrb{\mathcal B(\p P)} \defeq \mathcal B(\lrb{\p P}). \notag
\end{gather*}
A pure context $\p \Gamma = \p{x_1 : P_1, \ldots, x_n : P_n}$ is
interpreted as $\lrb{\p \Gamma} \defeq \lrb{\p P_1} \otimes \lrb{\p P_2} \otimes \cdots \otimes \lrb{\p P_n}$.  
Similarly, a mixed context $\Gamma = x_1 : A_1, \ldots, x_n : A_n$ is interpreted as 
$\lrb{\Gamma} \defeq \lrb{A_1} \otimes \lrb{A_2} \otimes \cdots \otimes \lrb{A_n}$.
Pure judgements $\p{\Gamma \vdash m : P}$ are interpreted as morphisms
$\lrb{\p{\Gamma \vdash m : P}} \colon \lrb{\p \Gamma} \rightarrow \lrb{\p P}$ in \textbf{FHilb}
and mixed judgements $\Gamma \vdash m : A$ are interpreted as morphisms
$\lrb{\Gamma \vdash m : A} \colon \lrb{\Gamma} \rightarrow \lrb{A}$ in \textbf{FC*}.  
The denotational interpretation of judgements from both fragments is given in Figure~\ref{fig:t_d}.  
We sometimes write $\lrb{M}$ as a shorthand for
$\lrb{\Gamma \vdash M : T}$. For simplicity, we often suppress some of the coherent natural isomorphisms related
to the monoidal structure (e.g. the $\alpha$ monoidal associator).
\begin{figure}[t]
  \begin{align*}
    \lrb{x : S \vdash x : S} &\defeq id \\
    \lrb{\Gamma \vdash M \Plus N : S} &\defeq \lrb{M} \Plus \lrb{N}\\
    \lrb{\Gamma \vdash \alpha \sdot M : S} &\defeq \alpha \sdot \lrb{M} \\
    \lrb{\p{\Gamma \vdash \topintro : \top}} &\defeq {!} \\
    \lrb{\Gamma, \Delta \vdash \elimzero(m) : A} &\defeq {!} \circ (\lrb{m} \otimes id)\\
    \lrb{\vdash \ast : \one} &\defeq id \\
    \lrb{\Gamma, \Delta \vdash \elimone(M,N) : S} &\defeq \lambda \circ (\lrb{M} \otimes \lrb{N}) \\
    \lrb{\Gamma \vdash \lambda \abstr{x}M : {S \multimap T}} &\defeq \Phi(\lrb{M}) \\
    \lrb{\Gamma, \Delta \vdash M~N : T} &\defeq \eval \circ (\lrb{M} \otimes \lrb{N}) \\
    \lrb{\p{\Gamma \vdash \pair{m}{n} :} \p P \pwith \p Q} &\defeq \pair{\lrb{\p{m}}} {\lrb{\p{n}}}\\
    \lrb{\p{\Gamma, \Delta \vdash \elimwith^i(m,\abstr{x}n) : R}} &\defeq \lrb{\p{n}}\circ \sigma \circ 
      (\pi_i \otimes id) \circ (\lrb{\p{m}} \otimes id) \\
    \lrb{\Delta \vdash \inl(m) : A \oplus B} &\defeq i_1 \circ \lrb{m} \\
    \lrb{\Delta \vdash \inr(n) : A \oplus B} &\defeq i_2 \circ \lrb{n} \\
    \lrb{\Gamma, \Delta \vdash \elimplus(m,\abstr{x}n,\abstr{y}o) \colon C } &\defeq [\lrb{n}, \lrb{o}] \circ d
      \circ \sigma \circ (\lrb{m} \otimes id) \\
    \lrb{\Gamma, \Delta \vdash M \otimes N : S_1 \otimes S_2} &\defeq \lrb{M} \otimes \lrb{N} \\
    \lrb{\Gamma, \Delta \vdash \elimtens(M , \abstr{xy}N) : T} &\defeq \lrb{N} \circ (id \otimes \lrb{M})\\
    \lrb{\vdash \mathcal B(t) : \mathcal B (\p P)} &\defeq \mathcal B (\lrb{t})\\
    \lrb{\Gamma,\Delta\vdash m~n : \mathcal B(\p Q)} &\defeq \mathcal B(\eval) \circ \tau \circ (\lrb{m} \otimes \lrb{n}) \\
    \lrb{ \Gamma \vdash \tau(m) : \mathcal B(\p P \p \otimes \p Q)} &\defeq \tau \circ \lrb{m}\\
    \lrb{\Delta, x : {S_2}, y : {S_1}, \Gamma \vdash M : {T}} &\defeq \lrb{M}  \circ (id \otimes \sigma \otimes id)
  \end{align*}
  \caption{Interpretation of judgements.}
  \label{fig:t_d}
\end{figure}

Our interpretation is sound with respect to the cut-elimination process (Theorem~\ref{thm:soundnesshook}), and complete (Theorem~\ref{thm:adequacy}) with respect to a notion of contextual equivalence (Definition~\ref{def:contextual_equivalence}).  
Since this completeness is established with respect to a different relation than the one used for soundness, the result is usually referred to as \emph{adequacy}.
The omitted proofs can be found in Appendix~\ref{app:denotational}.

\begin{restatable}[Soundness]{theorem}{soundnesscor}
  \label{thm:soundnesshook}
If $M \hookrightarrow^* N$ then $\lrb{M} = \lrb{N}$.  
\end{restatable}
\vspace{-1.8\baselineskip}\qed

\begin{definition}[Elimination context]
  An elimination context $E$ is a term with exactly one free variable, denoted by $[.]$, defined by the following grammar:
  \begin{align*}
    E & \defeq [.] 
    \mid \elimzero(E) 
    \mid \elimone(E, M) 
    \mid E~M \\
    & \mid \p{\elimwith^1(}E\p{, \abstr{x}m)} 
    \mid \p{\elimwith^2(}E\p{, \abstr{x}m)} 
    \mid \elimplus(E, \abstr{x}m, \abstr{y}n) 
    \mid \elimtens(E, \abstr{xy}M)
  \end{align*}
  The substitution of $[.]$ by a term $M$ in $E$ is denoted by $E[M]$.
\end{definition}

\begin{definition}[Contextual equivalence]
  \label{def:contextual_equivalence}
  Two terms $M$ and $N$ are contextually equivalent ($M \sim N$) if, for every
  elimination context $[.] \vdash E : T$, with $T \in \{\one, 1, \mathcal B(\one)\}$, there
  exists a value $V$ such that $E[M] \hookrightarrow^* V$ iff 
  $E[N] \hookrightarrow^* V$.
\end{definition}

\begin{restatable}[Adequacy]{theorem}{adequacy}
  \label{thm:adequacy}
  If $\lrb{\vdash M:T} = \lrb{\vdash N:T}$ then $M \sim N$.
\end{restatable}
\vspace{-1.8\baselineskip}\qed

\section{Encoding Quantum Computing}
\label{sssec:encoding}
The fact that our calculus is linear in a linear algebraic sense follows as a straightforward corollary of adequacy (Theorem~\ref{thm:adequacy}).  
\begin{corollary}[Linearity]
  \label{cor:linearity}
  Let $x : S \vdash M : T$,
  $\Delta \vdash N_1 : S$,
  $\Delta \vdash N_2 : S$,
  and $\Delta \vdash N : S$ (all in the same fragment), then
  $(\lambda x.M)~(\alpha\sdot N_1\Plus\beta\sdot N_2) \sim
  \alpha\sdot(\lambda x.M)~N_1 \Plus \beta\sdot (\lambda x.M)~N_2$.
\end{corollary}

Remark that Corollary~\ref{cor:linearity} excludes $\mathcal B$, since those terms are always closed.

Let $Q$ be the set of closed proof terms $\p m$ of $\one \pwith \one$, modulo
the equivalence relation $\equiv$. From now on, we write $\pqbit$ for $\one\pwith\one$.

\begin{theorem}[One-to-one correspondence~{\cite[Lemmas 3.7 and 3.8]{DiazcaroDowekMSCS24}}]
  \label{thm:v_space_type}
  The set $Q$ forms a finite-dimensional vector space, with vector addition and
  scalar multiplication given by $\Plus$ and $\sdot$, respectively.

  For every element $[\p m] \in Q$, there is a vector $\ul{[\p
  m]} \in \mathbb{C}^2$.  Conversely, for any vector $\vec{v} \in
  \mathbb{C}^2$, there is a closed term $\vdash \ov{\vec{v}} : \pqbit$.
  Moreover, this correspondence preserves the structure:
  $\ul{[\p{m_1 \Plus m_2}]} = \ul{[\p{m_1}]} + \ul{[\p{m_2}]}$
    and
  $\ul{[\p a \sdot \p m]} = \p a \cdot \ul{[\p m]}$.
    \qed
\end{theorem}

\begin{example}[Qubit encoding]
  A single qubit $a \ket{0} + b \ket{1} \in \mathbb{C}^2$ is encoded by the term  
  $\p{\langle a} \sdot \p{\ast , b} \sdot \p{\ast \rangle}$.
  In particular, we define the basis vectors as:
    $\ket{0} \defeq \p{\langle} \p{\ast , 0} \sdot \p{\ast \rangle}$
    and
    $\ket{1} \defeq \p{\langle 0} \sdot \p{\ast , } \p{\ast \rangle}$.
  An $n$-qubit is a vector in $\mathbb{C}^{2^n}$, and it is encoded as a linear
  combination of $n$-fold tensor products of 1-qubit encodings.
  As usual, we write $\ket{b_1\cdots b_n}$ for the $n$-qubit $\ket{b_1} \otimes \cdots \otimes \ket{b_n}$.
  For example, the
  $2$-qubit entangled state $\frac{\ket{00} + \ket{11}}{\sqrt{2}}$ is
  encoded as: $\vdash\tfrac{\p 1}{\p{\sqrt{2}}} \sdot \ket{00}
  \Plus \tfrac{\p 1}{\p{\sqrt{2}}} \sdot \ket{11}:\pqbit\p{\otimes}\pqbit$,
  that is,
  \(
    \vdash\tfrac{\p 1}{\sqrt{\p 2}} \sdot (\p{\langle} \p{\ast , 0} \sdot \p{\ast \rangle} \otimes \p{\langle} \p{\ast , 0} \sdot \p{\ast \rangle})
    \Plus
    \tfrac{\p 1}{\sqrt{\p 2}} \sdot (\p{\langle 0} \sdot \p{\ast , } \p{\ast \rangle} \otimes \p{\langle 0} \sdot \p{\ast , } \p{\ast \rangle}) : (\one\pwith\one) \p{\otimes} (\one\pwith\one)
  \).
\end{example}

\begin{remark}
In~\cite{DiazcaroDowekMSCS24}, $2^n$-dimensional vectors are encoded as proof terms of $\one^{\pwith{2^n}}$.  
Here we prefer to use the tensor notation, which provides finer control over term structure---in particular, allowing for  
a more direct encoding of constructs such as the Quantum Switch (see Example~\ref{ex:qs}).
\end{remark}

We first show how to encode $2 \times 2$ matrices (Theorem~\ref{lem:matrix_app}), and then generalise the construction to $n \times n$ matrices (Corollary~\ref{cor:tens_encoding}).

\begin{theorem}[$2 \times 2$ matrices~{\cite[Theorem 3.10]{DiazcaroDowekMSCS24}}]
\label{lem:matrix_app}
Let $M$ be a $2 \times 2$ complex matrix.  
Then there exists a closed proof term $\p f$ of $\pqbit \pmultimap \pqbit$ such that for any vector $\vec{v} \in \mathbb{C}^2$, we have:
  $M\vec{v} = \ul{\p f(\ov{\vec{v}})}$.
\end{theorem}
\begin{proof}  
  Let $M \defeq a \ketbra 00 + b \ketbra 01 + c \ketbra 10 + d \ketbra 11$ be a $2\times 2$ matrix, and let
  $\p f \defeq \p{\lambda \abstr{x} \elimwith^1(x, \abstr{x_1}{f_1 x_1}) \Plus \elimwith^2(x, \abstr{x_2}{f_2 x_2})}$,
  where
  $\p{f_1} \defeq \p{\lambda \abstr{x}\elimone(x, } \p{\langle a} \sdot \p{\ast , c} \sdot \p{\ast \rangle})$
  and
  $\p{f_2} \defeq \p{\lambda \abstr{x}\elimone(x, } \p{\langle b} \sdot \p{\ast , d} \sdot \p{\ast \rangle})$.
  Let $\vec{v} \defeq a_1\ket 0+b_1\ket 1 \in \mathbb{C}^2$.  
  Notice that
  $\p f(\ov{\vec{v}}) \hookrightarrow^* \p{f_1}(\p a_1 \sdot \p \ast) \Plus \p{f_2}(\p b_1 \sdot \p \ast)
  \hookrightarrow^* 
  \p{\langle} (\p a_1\p a+\p b_1\p b) \sdot \p{\ast ,} (\p a_1\p c +\p b_1\p d) \sdot \p{\ast \rangle}$.
  Thus, $M\vec{v} = \ul{\p f(\ov{\vec{v}})}$.
  \qed
\end{proof}

\begin{corollary}
\label{cor:tens_encoding}
Let $M$ be a linear transformation on $\mathbb C^{2^n}$.  Then there exists a
closed term $\p f$ such that: $\underline{\p f(\ov{\vec{v}})} = M\vec{v}$ for
all $\vec{v} \in \mathbb{C}^{2^n}$.
\end{corollary}
\begin{proof}
  Let $\p f_1$ and $\p g_1$ encode two matrices $M_1$ and $N_1$ respectively.
Thethe tensor product $M_1 \otimes N_1$ can be encoded by:
  $\p{\lambda z.\, \elimtens(z, \abstr{xy}{f_1 x \otimes g_1 y})}$.
Thus, the result follows from Theorem~\ref{lem:matrix_app} and the fact that any linear transformation  
on a finite-dimensional space $V \otimes W$ can be written as a linear combination of operators of the form $f \otimes g$,  
with $f : V \rightarrow V$ and $g : W \rightarrow W$.
\qed
\end{proof}

\begin{example}[CNOT gate]
  \label{ex:CNOT}
  The $CNOT$ gate can be written as a linear combination of tensor products of Pauli matrices $X$, $Z$, and the identity $I$:  
  \(
    \tfrac{1}{2}(I \otimes I + Z \otimes I + I \otimes X - Z \otimes X)
  \).
  Let $\p f_x$, $\p f_z$ be the proof terms representing $X$ and $Z$, respectively.  
  Then, by Corollary~\ref{cor:tens_encoding}, a proof term representing the $CNOT$ gate for  
  proposition $\pqbit \p\otimes \pqbit \p\multimap \pqbit \p\otimes \pqbit$ is:
  \\
 $\p{\lambda \abstr{z} \elimtens\Big(z, \abstr{xy}}
    \tfrac{\p 1}{\p 2} \sdot \p{(x \otimes y)} \Plus
    \tfrac{\p 1}{\p 2} \sdot \p{(f_z~x \otimes y)}
    \Plus 
    \tfrac{\p 1}{\p 2} \sdot \p{(x \otimes f_x~y)} \Plus
  \tfrac{\p{-1}}{\p 2} \sdot \p{(f_z~x \otimes f_x~y)\Big)}$.
\end{example}

\begin{example}[Quantum switch]
\label{ex:qs}
The quantum switch is a higher-order construction that applies two proofs $\p f,\p g : \p P\pmultimap\p P$ in opposite order depending on the value of a control qubit.  
It can be represented in \ocalc as the proof term:\\
$\p{\lambda \abstr{h} \lambda \abstr{z}
  \elimtens\Big(\!h, \abstr{fg}{
    \elimtens(z, \abstr{xy}{
      \elimwith^1(x,\abstr{x_1}{x_1 \otimes f(g\,y)})\! \Plus
      \elimwith^2(x,\abstr{x_2}{x_2 \otimes g(f\,y)})
    })
  }\!\Big)}$\\
proving the proposition 
$((\p{P \multimap P}) \p\otimes (\p{P \multimap P})) \p\multimap (\pqbit \p\otimes \p P) \p\multimap (\pqbit \p\otimes \p P)$.
\end{example}

The mixed-state fragment allows us to represent quantum measurements.
Let $m \colon M_2 \to M_2$ be the measurement map
\[
  \textit{m}(\rho) \defeq \textstyle\sum_{i=1}^n P_i \rho P_i^\dagger,
\]
determined by a choice of positive elements $P_i \in M_2$ that sum to the identity, i.e. a POVM.
Let $\p m_i$ denote the encoding in \ocalc\ of the linear map $P_i$. Then $\mathcal B(\p m_i)$ encodes
the map $P_i(\cdot)P_i^\dagger$ in a way that allows us to recover its action on positive elements
(which is sufficient considering our choice of morphisms).
Therefore the encoding of $m$ in $\ocalc$ is
\[
  \vdash\mathcal B(\p{m_1}) \Plus \mathcal B(\p{m_2}) \Plus \cdots \Plus \mathcal B(\p{m_n}):\mathcal B(\pqbit\pmultimap\pqbit) .
\]

This construction naturally generalises to measurements on finitely many
qubits by applying Corollary~\ref{cor:tens_encoding}.

\begin{example}[Measurement in the computational basis]
  \label{ex:meas}
  Consider the $2 \times 2$ projection matrices
  $\ketbra{0}{0}$ and $\ketbra{1}{1}$.  
  Let $\p{m_1}, \p{m_2}$ denote the encoding of these matrices in \ocalc.
  Then, the measurement of a single qubit in the computational basis is given by the
  term $\vdash \mathcal B(\p{m_1}) \Plus \mathcal B(\p{m_2}) : \mathcal B(\pqbit\pmultimap\pqbit)$.
  Consider, for example, the term $\p{\ket+} \defeq \p{\langle \frac{1}{\sqrt{2}}} \sdot \p{\ast, \frac{1}{\sqrt{2}}} \sdot \p{\ast \rangle}$.
  Then:
  \begin{align*}
    (\mathcal B(\p{m_1}) \Plus \mathcal B(\p{m_2}))\,\mathcal B(\p{\ket +})
    & \hookrightarrow \mathcal B(\p{m_1~\ket +}) \Plus \mathcal B(\p{m_2~\ket +})\\
  & \hookrightarrow^* \mathcal B\big(\p{\langle} \tfrac{\p 1}{\p{\sqrt{2}}} \sdot \p{\ast , 0} \sdot \p{\ast \rangle}\big)
\Plus \mathcal B\big(\p{\langle 0} \sdot \p{\ast , }\tfrac{\p 1}{\p{\sqrt{2}}} \sdot \p{\ast \rangle}\big).
  \end{align*}
  This is, as expected, the representation of the density matrix
  $\tfrac{1}{2} \cdot \ketbra{0}{0} + \tfrac{1}{2} \cdot \ketbra{1}{1}$.
\end{example}

\begin{example}[Bell measurement]
\label{ex:bell_m}
The Bell basis is given by $\{\beta_{00},\beta_{01},\beta_{10},\beta_{11}\}$, with $\beta_{ij} \defeq CNOT((H\ket i)\otimes\ket j)$.
Let $\p m_{ij}$ be the encoding of the projector $\ketbra{\beta_{ij}}{\beta_{ij}}$. Then, Bell measurement is represented by:
\[
  \vdash\mathcal B(\p m_{00}) \Plus \mathcal B(\p m_{01}) \Plus \mathcal B(\p m_{10}) \Plus \mathcal B(\p m_{11})
  : \mathcal B(\pqbit \p\otimes \pqbit \pmultimap \pqbit \p\otimes \pqbit)
\] \end{example}

\begin{example}[Teleportation]
\label{ex:tele}
Quantum teleportation transfers the state of a qubit $\ket{\psi}$ using shared entanglement and classical communication. Let $\p{f_i}$ be the proof terms representing the unitary corrections $U_i$, and $\p{h_i}$ the proof terms for the projectors $P_i$ used in a Bell measurement. Then, consider the proof term:
\begin{align*}
  \p{U} \defeq
  &\mathcal B(\p{\lambda \abstr{z} \elimtens(z, \abstr{xy} f_1~x \otimes h_1~y)}) \Plus
  \mathcal B(\p{\lambda \abstr{z} \elimtens(z, \abstr{xy} f_2~x \otimes h_2~y)}) \Plus\\
  &\mathcal B(\p{\lambda \abstr{z} \elimtens(z, \abstr{xy} f_3~x \otimes h_3~y)}) \Plus
  \mathcal B(\p{\lambda \abstr{z} \elimtens(z, \abstr{xy} f_4~x \otimes h_4~y)})
\end{align*}
of the proposition
$\mathcal B(\pqbit \p\otimes (\pqbit \p\otimes \pqbit)) \pmultimap \mathcal B(\pqbit \p\otimes (\pqbit \p\otimes \pqbit))$.
Teleportation is then encoded as the proof term for the proposition
$\mathcal B(\pqbit) \multimap \mathcal B(\pqbit \p\otimes \pqbit \p\otimes \pqbit)$ given by
\[
  \text{Telep} \defeq \lambda z.\, \p{U}~(\tau(\mathcal B(\beta_{00}) \otimes z)).
\]
Here, $\tau$ is used to encode the Bell state and the state to be teleported into a single density matrix.
\end{example}

\section{Conclusion}
\label{sec:conclusion}

We introduced a proof language for Intuitionistic Multiplicative Additive Linear Logic (IMALL), extended with a modality $\mathcal B$ to integrate both pure and mixed-state quantum computation in a unified setting. 
The language is equipped with a categorical model that serves two key roles in our development:
(1) the design of the logical system was inspired and largely extracted from the categorical/mathematical model;
(2) this model is relevant for the mathematical formulation of finite-dimensional quantum theory and we use it to justify design choices in the logical system.
Our logical system enables the expression of the pure-state quantum switch (Example~\ref{ex:qs}), which can, in principle, be applied in a mixed-state context subsequently.
This requires combining pure-state primitives with mixed-state primitives and cannot be easily achieved in most other logical/type systems.

Future work includes extending the model to full Intuitionistic Linear Logic (ILL) and clarifying the connection with Lambda-S, whose semantics is based on the same adjunction used to interpret the exponential modality of linear logic~\cite{DiazcaroMalherbeACS2020,DiazcaroMalherbeLMCS22}.

\paragraph{Acknowledgements.} We thank Cole Comfort, James Hefford, and Bert Lindenhovius for discussions. This work has been partially funded by the French National Research Agency (ANR) within the framework of ``Plan France 2030'', under the research projects EPIQ ANR-22-PETQ-0007, HQI-Acquisition ANR-22-PNCQ-0001 and HQI-R\&D ANR-22-PNCQ-0002, by the European Union through the MSCA SE project QCOMICAL (Grant Agreement ID: 101182520), and by the Uruguayan CSIC grant 22520220100073UD.

\bibliographystyle{splncs04}
\bibliography{refs}
\newpage
\appendix

\section{Omitted Proofs in \secref{sec:language}}
\label{app:proofs}
\Uniqueness*
\begin{proof}
  Proof is by induction on $M$:
  \begin{itemize}
    \item $M$ is $\p x$ or $x$: In these cases $\aform{M} = 1 \sdot M$.
    \item $M$ is $\p \topintro$: Again in this case $\aform{M} = 1 \sdot M$.
    \item $M$ is $\p\ast$ or $\ast$: In these cases, $\aform{M} = 1 \sdot M$. 
    \item $M$ is $\p{\elimone(m,n)}$: By induction hypothesis $\p m, \p n$ have unique algebraic forms $\aform{\p m}, \aform{\p n}$.
      We have, $\aform{M} = 1\sdot \p{\elimone(}\aform{\p m}\p ,\aform{\p n}\p)$.
    \item $M$ is $\p{\lambda \abstr{x}m}$: By induction hypothesis, $\p m$ has a algebraic form $\aform{\p m}$. We define,
      $\aform{M} = 1 \sdot \p{\lambda \abstr{x}}\aform{m}$. Since $\aform{\p m}$ is unique, so is $\aform{M}$.
    \item $M$ is $\p{m~n}$: We have, $\aform{\p{m~n}} = \p 1 \sdot \aform{\p m}~\aform{\p n}$.
    \item $M$ is $\p{m \otimes n}$: By the induction hypothesis, we have $\aform{\p m} = \textstyle\sum_{i=1}^m \p{a_i} \sdot \mathfrak{b}^i$,
      $\aform{\p n} = \textstyle\sum_{j=1}^n \p{a_i} \sdot \mathfrak{b}^j$. Hence from transitivity and contextual rules
      it follows that $M = (\textstyle\sum_{i=1}^m \p{a_i} \sdot \mathfrak{b}^i) \p \otimes (\textstyle\sum_{j=1}^n \p{a_i} \sdot \mathfrak{b}^j)$.
      Following the rules for rewriting under $=$ in Figure~\ref{fig:bidir_rules}, we have,
      $M = \textstyle\sum_{i=1}^m \textstyle\sum_{j=1}^n \p{a_ia_j} \sdot \p{\mathfrak{b}^i \otimes \mathfrak{b}^j}$. Now, renaming indices
      and rewriting if needed using the rule $\alpha \sdot M \Plus \beta \sdot M = (\alpha + \beta)\sdot M$, gives us the desired 
      result.
    \item $M$ is $\p{\elimtens(m, \abstr{xy}n)}$: In this case we have $\aform{M} = \p{1} \sdot \p{\elimtens(} \aform{\p m} \p{, \abstr{xy}} \aform{\p n}\p{)}$.
    \item $M$ is $\p{\pair{m}{n}}$: We have, $\aform{M} = \p 1 \sdot \p{\langle} \aform{\p m}\p , \aform{\p n}\p \rangle$.
    \item $M$ is $\p{\elimwith^i(m, \abstr{x}n)}$: Follows using the same argument as the cases above.
    \item $M$ is $\p{m \Plus n}$: By induction hypothesis we have $\aform{\p m} = \textstyle\sum_{i=1}^n \p{a_i}\sdot \mathfrak{b}^i$ and
      $\aform{\p n} = \textstyle\sum_{j=1}^k \p{a_j}\sdot \mathfrak{b}^j$. Hence, 
      $M = \textstyle\sum_{i=1}^n \p{a_i}\sdot \mathfrak{b}^i \Plus \textstyle\sum_{j=1}^k \p{a_j}\sdot \mathfrak{b}^j$. Now, rewriting using commutativity
      and the rule $\alpha \sdot M \Plus \beta \sdot M = (\alpha + \beta)\sdot M$ gives us the result.
    \item $M$ is $\p{a}\sdot \p m$: 
      We have $\aform{\p m} = \textstyle\sum_{i=1}^n \p{a_i}\sdot \mathfrak{b}^i$ by the induction hypothesis. Hence, 
      $M = \p a \sdot \textstyle\sum_{i=1}^n \p{a_i}\sdot \mathfrak{b}^i$. Now rewriting using distributivity, commutativity and the 
      rule $\alpha \sdot (\beta \sdot M) = \alpha \beta \sdot M$, gives us the result.
    \item $M$ is $\mathcal B(\p m)$: In this case we have, $\aform{M} = 1 \sdot \mathcal B(\aform{\p m})$.
    \item $M$ is $\tau(m)$: In this case we have, $\aform{M} = 1 \sdot \tau(\aform{m})$.
  \end{itemize}
  The remining cases for the mixed fragment follow using similar arguments as those above.
  \qed
\end{proof}

\progress*
\begin{proof} 
  The $\leftrightarrows$ is only used to rewrite terms when $\rightarrow$ gets
  stuck, if the reduct is not in its algebraic form. We prove progress by
  induction on $M$.  $M$ cannot be a variable since it is closed, and if If $M$
  is a value, we are done.
  Otherwise, we have the following cases:
  \begin{itemize}
    \item $M$ is $\p{m_1 \Plus m_2}$:  By induction hypothesis, either $\p{m_1}, \p{m_2}$ reduce further or they are values. If they reduce further then 
      so does $M$ by the contextual reduction rules. If they are values we do a case analysis based on the proposition $\p P$ of $M$. If 
      $\p P = \p{P_1 \with P_2}, \p{P_1 \multimap P_2}$, then $\p{m_1}, \p{m_2}$ are $ 
      \p{\pair{u_1}{u_2}}, \p{\pair{v_1}{v_2}}$ or $\p{\lambda \abstr{x}n_1}, \p{\lambda \abstr{x}n_2}$ respectively. In these cases $M$ reduces further under $\rightarrow$ 
      after a rewrite of $M$ to $\aform{M}$. Now, if
      $\p P = \one, \p{P_1 \otimes P_2}$, then they are of the form $\p \ast, \p a \sdot \p \ast$ or 
      $\textstyle\sum_{i=1}^n\p{a_i} \sdot \p{v_b}^i$ as in the syntax. In all these cases, $M$ does not reduce under 
      $\rightarrow$ and hence a rewrite of $M$ under $\leftrightarrows$ gives us the desired result. Hence, in all possible cases $M$ reduces under $\hookrightarrow$.

    \item $M$ is $\p a \sdot \p n$: Again by induction hypothesis, $\p n$ reduces further until it is a value, in which case so does $M$. Now we assume $\p n$ 
      is a value. We do a case analysis based on the proposition $\p P$ of $M$. 
      If $\p P = \p{P_1 \with P_2}, \p{P_1 \multimap P_2}$, then $\p n$ is $ 
      \p{\pair{v_1}{v_2}}$ or $\p{\lambda \abstr{x}n'}$ respectively. In these cases $M$ reduces further under $\rightarrow$ to take the form $\p v$ 
      defined in the syntax. Now, if
      $\p P = \one, \p{P_1 \otimes P_2}$, then it is of the form $\p \ast, \p a \sdot \p \ast$ or 
      $\textstyle\sum_{i=1}^n\p{a_i} \sdot \p{v_b}^i$ as in the syntax. In both these cases, $M$ does not reduce under
      $\rightarrow$ and hence a rewrite of $M$ under $\leftrightarrows$ gives us the desired result. Hence, in all possible cases $M$ reduces under $\hookrightarrow$.

    \item $M$ is $\p{\elimone(m_1,m_2)}$: By induction hypothesis, $\p{m_1}$ reduces until it is a value, in which case so does $M$. Now, we assume that $\p{m_1}$ is a value. It 
      follows that $\p{m_1}$ is $\p{\ast}$ or $\p a \sdot \p \ast$. 
      Hence, it follows that $\p m$ reduces further either through a potential rewrite of $\p{m_1}$ under $\leftrightarrows$ or by reduction under $\rightarrow$.

    \item $M$ is $\p{m_1~m_2}$: Again by induction hypothesis both $\p{m_1}, \p{m_2}$ reduce until they are values, in which case so does $M$. Now, we assume both 
      $\p{m_1}, \p{m_2}$ are values. Hence, by the  rules given in Figure~\ref{fig:derivation_rules} it follows that $\p{m_1}$ is $\p{\lambda \abstr{x}n}$. Therefore, 
      $M$ reduces under $\rightarrow$.

    \item $M$ is $\p{\pair{m_1}{m_2}}$: By induction hypothesis and contextual reduction rules it follows that either $M$ is of the form $\p{\pair{v_1}{v_2}}$ for values $\p{v_1}, \p{v_2}$ or it reduces further under $\rightarrow$. 

    \item $M$ is $\p{\elimwith^1(m_1, \abstr{x} m_2)}$ or $\p{\elimwith^2(m_1, \abstr{x} m_2)}$: By induction hypothesis, either $\p{m_1}$ reduces or it is a value. If 
      it reduces then so does $M$ by the contextual reduction rules. If it is a value, then by the  rules given in Figure~\ref{fig:derivation_rules} it 
      follows that $\p{m_1}$ is $\p{\pair{v_1}{v_2}}$. Hence, $M$ reduces further in both the cases.

    \item $M$ is $\p{m_1 \otimes m_2}$:  By induction hypothesis and contextual reduction rules it follows that either $M$ is of the form $\p{v_1 \otimes v_2}$ for values
      $\p{v_1}, \p{v_2}$ or it reduces further under $\rightarrow$. In the first case a rewrite using Lemma~\ref{lem:unique_alg} gives us a value of the form $\p v$ from the syntax.
    \item $M$ is $\p{\elimtens(m_1, \abstr{xy} m_2)}$: By induction hypothesis, $\p{m_1}$ reduces further until it is a value, in which case so does $M$. Now we assume 
      $\p{m_1}$ is a value. From the  rules in Figure~\ref{fig:derivation_rules} it follows
      that $\p{m_1}$ is of the form $\p{v_n^1 \otimes v_n^2}$ or a linear combination of such values. Hence, following a potential rewrite under $\leftrightarrows$, $M$
      reduces in both the cases under $\rightarrow$.

    \item $M$ is $m_1 \Plus m_2$: Again by induction hypothesis, $m_1, m_2$ reduce further until they are values, in which case so does $M$. Hence we assume that 
      both $m_1, m_2$ are values. Now this case follows using an argument similar to the above case and rewriting
      using the rules for $\leftrightarrows$ defined in Figure~\ref{fig:bidir_rules}.

    \item $M$ is $p \sdot n$: By induction hypothesis, $n$ reduces further until it is a value, in which case so does $M$. Now we assume $n$ is a value, then 
      $n$ takes one of the forms of $v$ from the syntax. In all
      the cases, depending on the form the value takes, $M$ reduces further under $\rightarrow$ through a rewrite under $\leftrightarrows$ (when $M$ is of proposition $A \multimap B$) 
      or it is itself equal a value.

    \item $M$ is $\elimzero(n)$: By induction hypothesis, $n$ reduces until it is a value. Hence, so does $M$ by the contextual rules. Now, we assume $n$ is a value. 
      But by induction hypothesis, $n$ can never be a value.

    \item $M$ is $\elimone(m_1,m_2)$: By induction hypothesis, either $m_1$ reduces or it is a value. If it reduces then so does $M$ by the contextual reduction rules.
      If it is a value, then by the  rules given in Figure~\ref{fig:derivation_rules} it follows that $m_1$ is either $\ast$, or $p \sdot \ast$. In all cases after a 
      potential rewrite to $\aform{m_1}$, $M$ reduces further.

    \item $M$ is $m_1~m_2$: Again by induction hypothesis $m_1, m_2$ reduce further until they are values, in which case so does $M$. Now, we assume both $m_1,m_2$ are values. 
      We do a case analysis on the proposition of $m_1$. 
      First we consider the case when $m_1$ is of the proposition $A \multimap B$. Then it follows that $m_1$ is $\lambda \abstr{x}n$, 
      In this $M$ reduces further.
      Now we consider the case when $m_1$ is of proposition $\mathcal B(\p P_1 \p \multimap \p P_2)$. Then it follows that $m_1$ is $\mathcal B(\p{m_1})$ or 
      $\textstyle\sum_{i=1}^np_i \sdot \mathcal B(\p{v_i})$. Similarly, $m_2$ is also of the form
      $\mathcal B(\p v)$, or $\textstyle\sum_{i=1}^np_i \sdot \mathcal B(\p{u_i})$. This gives us 4 possible cases
      for $M$. All cases reduce using the rules in Figure~\ref{fig:red_rightarrow} after potential rewrites under $\leftrightarrows$.

    \item $M$ is $\inl(n)$ or $\inr(n)$: By induction hypothesis, $n$ reduces further until it is a value, in which case so does $M$. Now we assume $n$ is a value. 
      If $n$ is of the form $v_b$, then in both the cases $M$ is a value.
      Otherwise, $n$ can be rewritten to its algebraic form as a linear combination of $v_b^i$s, so that $M$ reduces further.

    \item $M$ is $\elimplus(m, \abstr{x} n, \abstr{y} o)$: By induction hypothesis $m$ reduces further until it is a value, in which case so does $M$. Now we assume 
      $m$ is a value. Then by Lemma~\ref{lem:unique_alg} it
      follows that $m$ can be rewritten to the form $\textstyle\sum_{i=1}^n p_i\sdot v_b^i$. Hence, $M$ reduces further.

    \item $M$ is $m_1 \otimes m_2$:  Again by induction hypothesis, $m_1, m_2$ reduce further until they are values, in which case so does $M$. Now, we assume
      $m_1, m_2$ are values, in which case $M$ is of the form $v_1 \otimes v_2$ for values
      $v_1, v_2$. A potential rewrite under $\leftrightarrows$ gives us a value of the form $v$ from
      the syntax.
    \item $M$ is $\elimtens(m_1, \abstr{x} m_2)$: By induction hypothesis $m_1$ reduces further until it is value, in which case so does $M$. Now we assume 
      $m_1$ is a value. Then by Lemma~\ref{lem:unique_alg} it
      follows that $m_1$ can be rewritten to the form $\textstyle\sum_{i=1}^n p_i\sdot v_b^i$. Hence, $M$ reduces further.
    \item $M$ is $\mathcal B(\p m)$: Again by induction hypothesis, $\p m$ reduces further until it is a value, in which case so does $M$.
    \item $M$ is $\tau(n)$: By induction hypothesis, $n$ reduces further until it is a value, in which case so does $M$. Now, we assume that $n$ is a value. Then by 
      Lemma~\ref{lem:unique_alg} it
      follows that $n$ can be rewritten to the form $\textstyle\sum_{i=1}^n p_i\sdot v_b^i$. Hence, $M$ reduces further.
     \qed
  \end{itemize}
\end{proof}

\begin{lemma}[Substitution lemma]
  \label{lem:subst}
  If $\Gamma,x: T \vdash M : S$ and $\Delta \vdash N : T$, then $\Gamma , \Delta \vdash (N/x)M : S$.
  \qed
\end{lemma}
\begin{proof}
  By induction on the derivation of $M$.
  \begin{itemize}
    \item $\p{\Delta , x : P, y : Q, \Gamma \vdash m : R}$: This case follows directly from the induction hypothesis.
    \item $\p{\Gamma \vdash \topintro : \top}$:
      This case follows vacuously.
    \item $\p{\Gamma , \Delta \vdash \elimzero(m') : P}$:
      We consider two cases. First case is when the variable being substituted is in $\p{\Gamma}$. Then the claim follows by using the induction
      hypothesis on the term $\p m'$ and the observation that $\p{(n/x)\elimzero(m')} = \p{\elimzero((n/x)m')}$. In the other case, the claim follows
      since $\p{(n/x)\elimzero(m')} = \p{\elimzero(m')}$.
    \item $\p{x : P \vdash x : P}$: In this case, $\p \Gamma = \emptyset$, $\p P = \p Q$ and $ \p{(n/x) m} = \p n$ and the claim follows.
    \item $\p{\vdash \ast : \one}$: This case follows vacuously.
    \item $\p{\Gamma, \Delta \vdash \elimone( m_1, m_2) : P}$: We consider two cases, one when the variable appears in $\p{m_1}$ and the 
      other when it appears in $\p{m_2}$. Both follow by application of induction hypothesis on $\p{m_1}, \p{m_2}$ respectively
      and the observation that in one case $\p{(n/x)\elimone(m_1, m_2)} = \p{\elimone((n/x)m_1, m_2)}$ and in the other
      case $\p{(n/x)\elimone(m_1, m_2)} = \p{\elimone(m_1, (n/x)m_2)}$.
    \item $\p{\Gamma \vdash \lambda \abstr{x}m' : P \multimap Q}$:  This case follows using the induction hypothesis
      on $\p m'$ and the observation that $\p{(n/x)\lambda \abstr{y} m'} = \p{\lambda \abstr{y}(n/x)m'}$. 
    \item $\p{\Gamma, \Delta \vdash m_1~m_2 : Q}$: Again we consider two cases, one when the variable appears in $\p{m_1}$ and the other when
      the variable appears in $\p{m_2}$. Now the claim follows by applying induction hypothesis on $\p{m_1}, \p{m_2}$ in both
      the cases respectively and the observation that in one case $\p{(n/x)(m_1~m_2)} = \p{((n/x)m_1)~m_2}$ and in the other
      $\p{(n/x)(m_1~m_2)} = \p{m_1~((n/x)m_2)}$.
    \item $\p{\Gamma \vdash \pair{m_1}{m_2} : P \with Q}$: Follows by application of induction hypothesis on $\p{m_1}, \p{m_2}$ respectively
      and the observation that $\p{(n/x)\pair{m_1}{m_2}} = \p{\pair{(n/x)m_1}{(n/x)m_2}}$.
    \item $\p{\Gamma, \Delta \vdash \elimwith^1(m_1,\abstr{x}m_2) : R}$ or $\p{\Gamma, \Delta \vdash \elimwith^2(m_1,\abstr{x}m_2) : R}$: Again we consider two cases, 
      one when the variable appears in $\p{m_1}$ and the other when it appears in $\p{m_2}$. The proof follows by 
      applying the induction hypothesis on $\p{m_1}, \p{m_2}$ respectively and the observation that 
      $\p{(n/x)\elimwith^i(m_1, \abstr{x}m_2)}$ is equal to $\p{\elimwith^i((n/x)m_1, \abstr{x}m_2)}$ in one case and
      $\p{(n/x)\elimwith^i(m_1, \abstr{x}m_2)}$ is equal to $\p{\elimwith^i(m_1, (n/x)\abstr{x}m_2)}$ in the other.
    \item $\p{\Delta, \Gamma \vdash m_1 \otimes m_2 : P \otimes Q}$: We consider two cases, one when the variable $\p x$ appears in $\p{m_1}$ and the other when $\p x$ appears
      in $\p{m_2}$ respectively. We apply induction hypothesis on the terms $\p{m_1}, \p{m_2}$ in both the cases. The claim follows using the
      observation that $\p{(n/x)(m_1 \otimes m_2)} = \p{(n/x) m_1 \otimes m_2}$ in one case and $\p{(n/x)(m_1 \otimes m_2)} = \p{m_1 \otimes (n/x)m_2}$
      in the other.
    \item $\p{\Delta, \Gamma \vdash \elimtens(m_1,\abstr{x y}m_2) : R}$: We consider two cases, one when the variable appears in $\p{m_1}$ and the other when it
      appears in $\p{m_2}$. In both the cases we apply the induction hypothesis on the terms $\p{m_1}, \p{m_2}$ respectively. Then again, the claim
      follows from the observation that $\p{(n/x)\elimtens(m_1, \abstr{yz}m_2)}$ is equal to $\p{\elimtens((n/x)m_1, \abstr{yz}m_2)}$ in one case and
      $\p{(n/x)\elimtens(m_1, \abstr{yz}m_2)}$ is equal to $\p{\elimtens(m_1, (n/x)\abstr{yz}m_2)}$ in the other.
    \item $\p{\Gamma \vdash a} \sdot \p{m' : P}$: This case follows from the induction hypothesis applied to $\p{m'}$ and the observation
      that $\p{(n/x)(a} \sdot \p{m')} = \p{a} \sdot \p{(n/x)m'}$.
    \item $\p{\Gamma \vdash m \Plus n : P}$: This case follows from the induction hypothesis applied to $\p{m_1}, \p{m_2}$ and the 
      observation that $\p{(n/x)(m_1 \Plus m_2)} = \p{(n/x)m_1 \Plus (n/x)m_2}$.
    \item $\Delta , x : B, y : A, \Gamma \vdash m : C$: This case follows directly from the induction hypothesis.
    \item $\vdash \ast : 1$: This case follows vacuously.
    \item $\Gamma, \Delta \vdash \elimone(m_1,m_2) : A$: We consider two cases, one when the variable appears in $m_1$ and the
      other when it appears in $m_2$. Both follow by application of induction hypothesis on $m_1, m_2$ respectively
      and the observation that in one case $(n/x)\elimone(m_1, m_2) = \elimone((n/x)m_1, m_2)$ and in the other
      case $(n/x)\elimone(m_1, m_2) = \elimone(m_1, (n/x)m_2)$.
    \item $\Delta \vdash \lambda \abstr{x}m' : A \multimap B$: We consider two cases. One when the variable being substituted is $x$ which
      is different from the one being abstracted as a lambda abstraction, and the other case when
      the variable being substituted is the one that is a part of the lambda abstraction. The latter case
      follows by the induction hypothesis applied on $m'$. The former case follows using the induction hypothesis
      on $m'$ and the observation that $(n/x)\lambda \abstr{y} m' = \lambda \abstr{y}(n/x)m'$.
    \item $\Gamma, \Delta \vdash m_1~m_2 : B$ or $\Gamma, \Delta \vdash m_1~m_2 : \mathcal B(\p Q)$: Again we consider two cases, 
      one when the variable appears in $m_1$ and the other when
      the variable appears in $m_2$. Now the claim follows by applying induction hypothesis on $m_1, m_2$ in both
      the cases respectively and the observation that in one case $(n/x)(m_1~m_2) = ((n/x)m_1)~m_2$ and in the other
      $(n/x)(m_1~m_2) = m_1~((n/x)m_2)$. The other case follows vacuously since in this case both $m_1, m_2$ are closed. 
    \item $\Delta, \Gamma \vdash m_1 \otimes m_2 : A \otimes B$: We consider two cases, one when the variable $\p x$ appears in $m_1$ and the other when $x$ appears
      in $m_2$ respectively. We apply induction hypothesis on the terms $m_1, m_2$ in both the cases. The claim follows using the
      observation that $(n/x)(m_1 \otimes m_2) = (n/x) m_1 \otimes m_2$ in one case and $(n/x)(m_1 \otimes m_2) = m_1 \otimes (n/x)m_2$
      in the other.
    \item $\Delta, \Gamma \vdash \elimtens(m_1,\abstr{x y}m_2) : C$: We consider two cases, one when the variable appears in $m_1$ and the other when it
      appears in $m_2$. In both the cases we apply the induction hypothesis on the terms $m_1, m_2$ respectively. Then again, the claim
      follows from the observation that $(n/x)\elimtens(m_1, \abstr{yz}m_2)$ is equal to $\elimtens((n/x)m_1, \abstr{yz}m_2)$ in one case and
      $(n/x)\elimtens(m_1, \abstr{yz}m_2)$ is equal to $\elimtens(m_1, (n/x)\abstr{yz}m_2)$ in the other.
    \item $\Delta \vdash \inl(m') : A \oplus B$ or $\Delta \vdash \inr(m') : A \oplus B$: This case follows by applying the induction hypothesis 
      on $m'$ and the observation that $(n/x)\inl(m') = \inl((n/x)m')$.
      The case for $\inr(m')$ follows similarly.
    \item $\Delta , \Gamma \vdash \elimplus(m_1,\abstr{x}m_2,\abstr{y}m_3) : C$: We consider two cases, one when the variable appears in $m_1$ and the other when it
      appears in the terms $m_2, m_3$. Both the cases follow by applying the induction hypothesis on the terms $m_1, m_2,m_3$ respectively and the
      observation that $(n/x)\elimplus(m_1, \abstr{y} m_2, \abstr{z} m_3)$ is equal to $\elimplus((n/x)m_1, \abstr{y} m_2, \abstr{z} m_3)$ in one case and
      $(n/x)\elimplus(m_1 , \abstr{y} m_2, \abstr{z} m_3)$ is equal to $\elimplus(m_1 , (n/x) \abstr{y} m_2, (N/x) \abstr{z} m_3)$ in the other.
    \item $\vdash \mathcal B(\p{m'}): \mathcal B(\p P)$: This case follows vacuously since $\p{m'}$ is a closed term.
    \item $\Gamma \vdash \tau(m'): \mathcal B(\p P \p \otimes \p Q)$: This case follows from the induction hypothesis and the observation that $(n/x)\tau(m') = \tau((n/x)m')$.
    \item $\Gamma \vdash p \sdot m' : A$: This case follows from the induction hypothesis applied to $m'$ and the observation
      that $(n/x)(p \sdot m') = p \sdot (n/x)m'$.
    \item $\Gamma \vdash m_1 \Plus m_2 : A$: This case follows from the induction hypothesis applied to $m_1, m_2$ and the
      observation that $(n/x)(m_1 \Plus m_2) = (n/x)m_1 \Plus (n/x)m_2$.
      \qed
  \end{itemize}
\end{proof}

\subred*
\begin{proof}
  Note that from the rules for $\leftrightarrows$ in Figure~\ref{fig:bidir_rules} and the rules sum and prod$(\alpha)$ from Figure~\ref{fig:derivation_rules} it follows that propositions are preserved by $\leftrightarrows$.
  Hence, it suffices to prove that $\rightarrow$ preserves propositions. We prove this using induction
  on the relation $\rightarrow$. We focus on the basic reduction rules in Figure~\ref{fig:red_rightarrow} and omit the proof for
  contextual rules since they follow easily from the induction hypothesis.
  \begin{itemize}
    \item $M=\p{\elimone( a} \sdot \p{\ast , n)}$ and $N=\p a \sdot \p n$.
      Since we have that $\p{\vdash \ast : \one}$, if we have $\p{\Gamma \vdash n : P}$,
      it follows from the derivation rules that $\p{\Gamma \vdash a} \sdot \p{n : P}$.
    \item $M=\elimone(p \sdot \ast, n)$ and $N=a \sdot n$.
      Since we have that $\vdash a \sdot \ast : 1$, if we have $\cdot \vdash n : A$,
      it follows from the derivation rules that $\cdot \vdash a\sdot n : A$.
    \item $M=\p{(\lambda \abstr{x}m) v}$ and $N=\p{(m_2/x)m_1}$, and the result
      holds by Lemma~\ref{lem:subst}.
    \item $M=(\lambda \abstr{x}m)~v$ and $N=(v/x)m$. The result
      holds by Lemma~\ref{lem:subst}.
    \item $M=\p{\elimwith^1(\pair{v_1}{v_2} , \abstr{x}n)}$ and
      $N=\p{(v_1/x)n}$. The result holds by
      Lemma~\ref{lem:subst}.
      The case for $\p{\elimwith^2(\pair{v_1}{v_2} , \abstr{x}n)}$ is similar.
    \item $M=\elimplus(\inl(v), \abstr{x}m, \abstr{y}n)$ or $M=\elimplus(\inr(v), \abstr{x}m, \abstr{y}n)$ and $N=(v/x)/m$
      or $N=(v/y)n$ respectively. The result holds in both the cases by an application of Lemma~\ref{lem:subst}.
    \item $M=\p{\elimtens(v_n^1 \otimes v_n^2, \abstr{x y}m')}$ and $N=\p{(v_n^1/x , v_n^2/y)m'}$. The result follows by a double application of Lemma~\ref{lem:subst}.
    \item $M=\elimtens(v_n^1 \otimes v_n^2, \abstr{xy}n)$ and $N=(v_n^1/x , v_n^2/y)n$. The result
      holds by a double application of Lemma~\ref{lem:subst}.
    \item $M=\p{a}\sdot \p{\topintro}$ and $N=\p{\topintro}$. The claim follows from the rules $\top_i^p$, sum$^p$,
      and prod$(a)^p$.
    \item $M=\textstyle\sum_{i=1}^n\p{a_i} \sdot \p{\lambda \abstr{x}m_i}$ and $N=\p{\lambda \abstr{x}}\textstyle\sum_{i=1}^n\p{a_i}\sdot \p{m_i}$. 
      The claim follows from the rules $\pmultimap_i^p$, sum$^p$, prod$(a)^p$.
    \item $M=\textstyle\sum_{i=1}^n{p_i} \sdot \lambda\abstr{x}m_i$ and $N=\lambda\abstr{x}\textstyle\sum_{i=1}^n{p_i} \sdot{m_i}$.
      The claim follows from the rules $\multimap_i^m$, sum$^m$, prod$(p)^m$.
    \item $M=\textstyle\sum_{i=1}^n \p{a_i}\sdot \p{\pair{u_{i1}}{u_{i2}}}$ and $N=\p{\langle} \textstyle\sum_{i=1}^n \p{a_i} \sdot \p{u_{i1},} \textstyle\sum_{i=1}^n \p{a_i} \sdot \p{u_{i2} \rangle}$.
      The claim follows from the rules $\pwith_i^p$, sum$^p$ and prod$(a)^p$.
    \item $M=\elimplus(\textstyle\sum_{i=1}^n p_i\sdot v_{bi},x.o,y.s)$ and $N=\textstyle\sum_{i=1}^n p_i\sdot \elimplus(v_{bi},x.o,y.s)$. The claim follows from the rules $\oplus_e^m$, sum$^m$, prod$(p)^m$.
    \item $M=\p{\elimtens(} \textstyle\sum_{i=1}^n \p{a_i} \sdot \p{v_b^i, \abstr{xy}n)}$ and $N=\textstyle\sum_{i=1}^n\p{a_i} \sdot \p{\elimtens(v_b^i, \abstr{xy}n)}$.
      The claim follows from the rules $\p \otimes_e^p$, sum$^p$, prod$(a)^p$.
    \item $M=\inl(\textstyle\sum_{i=1}^n p_i \sdot v_{bi})$ and $N=\textstyle\sum_{i=1}^n p_i \sdot \inl(v_{bi})$. The claim follows from the rules $\oplus_e^{i1}$, sum$^p$ and prod$(p)^p$. The case for $\inr(\textstyle\sum_{i=1}^n p_i \sdot v_{bi})$ is similar.
    \item $M=\tau(\mathcal B(\p v_1) \otimes \mathcal B(\p v_2))$ and $N=\mathcal B(\p{v_1 \otimes v_2})$. The claim follows from
      the rules $\mathcal B(\otimes)^m$, $\p \otimes_i^p$ and $\mathcal B^m$.
    \item $M=\tau(\textstyle\sum_{i=1}^n p_i \sdot v_{bi})$ and $N=\textstyle\sum_{i=1}^n p_i \sdot \tau(v_{bi})$. The claim follows from the rules $\mathcal B(\otimes)^m$, sum$^m$ and prod$(p)^m$.
    \item $M=\mathcal B(\p a \sdot \p \ast)$ and $N=|\p a|^2 \sdot \mathcal B(\p{\ast})$. The claim follows by applying the rules
      $\p \one_i^p$, $\mathcal B^m$, prod$(a)^p$ and prod$(|a|^2)^m$.
    \item $M=(\textstyle\sum_{i=1}^n p_i \sdot \mathcal B(\p{v_i}))~(\textstyle\sum_{j=1}^k q_j \sdot \mathcal B(\p{w_j}))$ and $N=\textstyle\sum_{i,j} p_i q_j \sdot \mathcal B(\p{v_i~w_j})$. The claim follows from the rules $\mathcal B(\pmultimap)_e^m$, sum$^m$ and prod$(p)^m$.
  \qed
  \end{itemize}
\end{proof}

\subsection{Proof of Strong Normalisation (Theorem~\ref{cor:strong-norm})}
\label{app:SN}
In this section, we prove that the reduction relation $\hookrightarrow$ is strongly normalising.  
To this end, we first extend the relation with additional rules, and show that strong normalisation holds for this extended relation.  
Since the original relation is included in the extended one, this immediately implies that $\hookrightarrow$ is also strongly normalising.

The following rules are added to form the extended relation:
\begin{align*}
  M \Plus N &\hookrightarrow M\\
  M \Plus N &\hookrightarrow N\\
  \alpha \sdot M &\hookrightarrow M
\end{align*}

The \emph{length} of a strongly terminating term $M$, denoted $\ell(M)$, is the maximal length of any reduction sequence starting from $M$.

\begin{lemma}
  \label{lem:plus_sn}
  If $M$ and $N$ are strongly normalising, then so is $M \Plus N$.
\end{lemma}

\begin{proof}
  We show that all one-step reducts of $M \Plus N$ are strongly normalising,  
  by induction on $\ell(M) + \ell(N)$, and then by structural induction on the size of the term.

  If the reduction applies to either $M$ or $N$, the claim follows directly from the induction hypothesis.  
  Otherwise, the reduction applies to the outermost $\Plus$. In that case, the possible reductions are:
  \begin{align*}
    \lambda \abstr{x}m \Plus \lambda \abstr{x}n &\hookrightarrow \lambda \abstr{x}(1 \sdot m \Plus 1 \sdot n)\\
    \p{\lambda \abstr{x}m \Plus \lambda \abstr{x}n} &\hookrightarrow \p{\lambda \abstr{x}(1} \sdot \p{m \Plus 1} \sdot \p{n)}\\
    \p{a} \sdot \p{\pair{m_1}{m_2}} \Plus \p{b} \sdot \p{\pair{m_1'}{m_2'}} &\hookrightarrow \p{\langle a} \sdot \p{m_1 \Plus b} \sdot \p{m_1', a} \sdot \p{m_2 \Plus b} \sdot \p{m_2' \rangle}\\
    M \Plus N &\hookrightarrow M\\
    M \Plus N &\hookrightarrow N
  \end{align*}
  In all cases, the reduct either decreases the size of the term or is built from strongly normalising subterms.  
  Hence, each reduct is strongly normalising by the induction hypothesis or by assumption.  
  \qed
\end{proof}

\begin{lemma}
  \label{lem:dot_sn}
  If $M$ is strongly normalising, then so is $\alpha \sdot M$.
\end{lemma}

\begin{proof}
  As in the previous lemma, we prove that all one-step reducts of $\alpha \sdot M$ are strongly normalising,  
  by induction on $\ell(M)$, and then on the size of the term.

  If the reduction applies to $M$, the claim follows from the induction hypothesis.  
  Otherwise, the reduction applies to the outermost scalar multiplication. The possible reductions are:
  \begin{align*}
    p \sdot \lambda \abstr{x}m &\hookrightarrow \lambda \abstr{x}(p \sdot m)\\
    \p{a} \sdot \p{\lambda \abstr{x}m} &\hookrightarrow \p{\lambda \abstr{x}(a} \sdot \p{m)}\\
    \p{a} \sdot \p{\topintro} &\hookrightarrow \p{\topintro}\\
    \p{a} \sdot \p{\pair{m_1}{m_2}} &\hookrightarrow \p{\langle a} \sdot \p{m_1, a} \sdot \p{m_2 \rangle}
  \end{align*}
  In all cases, either the term decreases in size, or it is constructed from strongly normalising subterms.  
  Thus, the result follows by induction or by the assumption that $M$ is strongly normalising.  
  \qed
\end{proof}

We let $\mathbf{SN}$ denote the set of proof terms that terminate in finitely many steps.

\begin{definition}
We inductively define the interpretation $\sn{T}$ of each proposition $T$ as the set of strongly normalising proof terms of $T$, as follows:
\begin{align*}
  \sn{\p\top} &\defeq \{ \p m \in \mathbf{SN} \} \\
  \sn{\one} &\defeq \{ \p m \in \mathbf{SN} \} \\
  \sn{\p P \pmultimap \p Q} &\defeq \left\{ \p m \in \mathbf{SN} \mid \p m \hookrightarrow^* \lambda \abstr{x}{n} \Rightarrow \forall \p v \in \sn{\p P},\ (\p v / x)n \in \sn{\p Q} \right\} \\
  \sn{\p P \pwith \p Q} &\defeq \left\{ \p m \in \mathbf{SN} \mid \p m \hookrightarrow^* \pair{\p{v_1}}{\p{v_2}} \Rightarrow \p{v_1} \in \sn{\p P} \ \text{and}\ \p{v_2} \in \sn{\p Q} \right\} \\
  \sn{\p{P \otimes Q}} &\defeq \left\{ \p m \in \mathbf{SN} \mid \p m \hookrightarrow^* \p{v_1} \otimes \p{v_2} \Rightarrow \p{v_1} \in \sn{\p P} \ \text{and}\ \p{v_2} \in \sn{\p Q} \right\} \\
  \sn{\zero} &\defeq \{ m \in \mathbf{SN} \} \\
  \sn{1} &\defeq \{ m \in \mathbf{SN} \} \\
  \sn{A \multimap B} &\defeq \left\{ m \in \mathbf{SN} \mid m \hookrightarrow^* \lambda \abstr{x}{m'} \Rightarrow \forall v \in \sn{A},\ (v / x)m' \in \sn{B} \right\} \\
  \sn{A \oplus B} &\defeq \left\{ m \in \mathbf{SN} \mid \begin{array}{l}
      m \hookrightarrow^* \inl(v) \Rightarrow v \in \sn{A},\\
      m \hookrightarrow^* \inr(v) \Rightarrow v \in \sn{B}
    \end{array}
  \right\} \\
  \sn{A \otimes B} &\defeq \left\{ m \in \mathbf{SN} \mid m \hookrightarrow^* v \otimes w \Rightarrow v \in \sn{A} \ \text{and}\ w \in \sn{B} \right\} \\
  \sn{\mathcal B(\p P)} &\defeq \left\{ m \in \mathbf{SN} \mid m \hookrightarrow^* \mathcal B(\p v) \Rightarrow \p v \in \sn{\p P} \right\}
\end{align*}
\end{definition}

\begin{definition}[Definitions and notations]
~
\begin{itemize}
  \item A \emph{substitution function} is a map $\theta : \s{0.85}{\versalita{Vars}} \rightarrow \s{0.85}{\versalita{Pure Values}} \cup \s{0.85}{\versalita{Mixed Values}}$.
  Given a term $M$, we write $\theta(M)$ for the result of replacing all free variables in $M$ by their image under $\theta$.
  
  \item Given a context $\Gamma$ and a substitution function $\theta$, we say that $\theta$ is \emph{valid for} $\Gamma$, written $\theta \vDash \Gamma$,  
  if for all $x : T \in \Gamma$, we have $\theta(x) \in \sn{T}$.
  
  \item For a term $M$, we define the set of one-step reducts of $M$ as:
  \[
    \mathit{Red}(M) \defeq \{\, N \mid M \hookrightarrow N \,\}
  \]
  
  \item A term is said to be \emph{neutral} if the last rule used in its derivation is an elimination rule, or the rule $\mathcal B(\otimes)^m$.
\end{itemize}
\end{definition}

\begin{lemma}[CR3]
  \label{lem:forward}
  Let $M$ be a neutral proof term of $T$.  
  If $N \in \sn{T}$ for all $N \in \mathit{Red}(M)$, then $M \in \sn{T}$.
\end{lemma}
\begin{proof}
Since all one-step reducts of $M$ are strongly normalising, all reduction sequences originating from $M$ are finite.  
Thus, $M \in \mathbf{SN}$. Moreover, since $M$ is neutral, it is not a value.

We proceed by case analysis on the proposition $T$:

\begin{itemize}
  \item $T \in \{\one, 1, \p\top, \p\zero\}$: Immediate, as $M$ strongly terminates by assumption and no further structure is required.

  \item $T = \p{P \multimap Q}$ (or $T = A \multimap B$):  
    We show the pure case; the mixed case is analogous.  
    Let $M \hookrightarrow N \hookrightarrow^* \lambda \abstr{x}{m}$.  
    Since $N \in \sn{T}$, we know that for all $\p v \in \sn{\p P}$, we have $(\p v / x)m \in \sn{\p Q}$.  
    Thus, by definition of $\sn{\p{P \multimap Q}}$, we conclude $M \in \sn{T}$.

  \item $T = \p{P \with Q}$:  
    In this case, $M \hookrightarrow N \hookrightarrow^* \pair{\p{v_1}}{\p{v_2}}$.  
    Since $N \in \sn{T}$, we know that $\p{v_1} \in \sn{\p P}$ and $\p{v_2} \in \sn{\p Q}$.  
    Therefore, $M \in \sn{\p{P \with Q}}$.

  \item $T = A \oplus B$:  
    In this case, $M \hookrightarrow N \hookrightarrow^* \inl(v)$ or $\inr(v)$.  
    Since $N \in \sn{T}$, we have either $v \in \sn{A}$ or $v \in \sn{B}$, respectively.  
    Hence, $M \in \sn{A \oplus B}$.

  \item $T = \p{P \otimes Q}$ (or $T = A \otimes B$):  
    Again, we show the pure case.  
    Let $M \hookrightarrow N \hookrightarrow^* \p{v_1} \otimes \p{v_2}$.  
    Since $N \in \sn{T}$, we know $\p{v_1} \in \sn{\p P}$ and $\p{v_2} \in \sn{\p Q}$.  
    Thus, $M \in \sn{\p{P \otimes Q}}$.

  \item $T = \mathcal B(\p P)$:  
    In this case, $M \hookrightarrow N \hookrightarrow^* \mathcal B(\p v)$.  
    Since $N \in \sn{T}$, we have $\p v \in \sn{\p P}$.  
    Hence, $M \in \sn{\mathcal B(\p P)}$.
    \qed
\end{itemize}
\end{proof}

\begin{lemma}[Adequacy of sum]
  \label{lem:sn_plus}
  If $M, N \in \textbf{SN}(T)$, then $M \Plus N \in \textbf{SN}(T)$.
\end{lemma}
\begin{proof}
  Proof is again by induction on the proposition $T$. Again by Lemma~\ref{lem:plus_sn} we have that $M \Plus N$ always terminates.
  \begin{itemize}
    \item $T$ is $\one$ or $1$:
      Follows directly from Lemma~\ref{lem:plus_sn}.
    \item $T$ is $\p P \pwith \p Q$:
      Since $M \Plus N$ terminates and $T$ is $\p{P \with Q}$, we have $M \Plus N \hookrightarrow^{\ast} \p{\pair{v_1}{v_2}}$.
      Hence, either $M \hookrightarrow^{\ast} \p{\pair{u_1}{u_2}}$, $N \hookrightarrow^{\ast} \p{\pair{u_1'}{u_2'}}$ with
      $\p{u_1 \Plus u_1'} \hookrightarrow^{\ast} \p{v_1}$ and $\p{u_2 \Plus u_2'} \hookrightarrow^{\ast} \p{v_2}$ or $M \hookrightarrow^{\ast} \p{\pair{v_1}{v_2}}$
      or $N \hookrightarrow^{\ast} \p{\pair{v_1}{v_2}}$. The first case follows by the induction hypothesis and the assumption that $M, N \in \sn{T}$.
      The other two cases follow from the assumption $M, N \in \sn{T}$.
    \item $T$ is $\p P_1 \pmultimap \p P_2$ or $A \multimap B$:
      Since both the cases are similar, we present only one case here. Since $M\Plus N$ terminates and $T$ is $\p{P \multimap Q}$, we have
      $M \plus N \hookrightarrow^{\ast} \p{\lambda \abstr{x}m}$. Now, either $M \hookrightarrow \p{\lambda \abstr{x}m_1}$, $M \hookrightarrow \p{\lambda \abstr{x}m_2}$
      and $\p m$ is of the form $\p{m_1 \Plus m_2}$ or $M \hookrightarrow^{\ast} \p{\lambda \abstr{x}m}$ or $N \hookrightarrow^{\ast} \p{\lambda \abstr{x}m}$.
      The first case follows by an application of the induction hypothesis to $\p{(v/x)m_1 \Plus m_2}$ and the assimption that $M,N \in \sn{T}$. The other two 
      cases follow from the assumption $M,N \in \sn{T}$.
    \item $T$ is $\p{P \otimes Q}$ or $A \otimes B$:
      We present only one case here, the other case follows using a similar argument. Since, $M \Plus N$ terminates and $T$ is $\p{P \otimes Q}$ we have
      $M \Plus N \hookrightarrow^{\ast} \p{v_1 \otimes v_2}$. Then either $M \hookrightarrow^{\ast} \p{v_1 \otimes v_2}$ or $N \hookrightarrow^{\ast} \p{v_1 \otimes v_2}$.
      Both cases follow from the assumption that $M,N \in \sn{T}$.
    \item $T$ is $A \oplus B$:
      Since, $M \Plus N$ terminates and $T$ is $A \oplus B$ we have
      $M \Plus N \hookrightarrow^{\ast} \inl(v)$(or $\inr(v)$). Then either $M \hookrightarrow^{\ast} \inl(v)$(or $\inr(v)$) or $N \hookrightarrow^{\ast} \inl(v)$(or $\inr(v)$).
      Both cases follow from the assumption that $M,N \in \sn{T}$.
    \item $T$ is $\mathcal B(\p P)$:
      Since, $M \Plus N$ terminates and $T$ is $\mathcal B(\p P)$ we have
      $M \Plus N \hookrightarrow^{\ast} \mathcal B(\p v)$. Then either $M \hookrightarrow^{\ast} \mathcal B(\p v)$ or $N \hookrightarrow^{\ast} \mathcal B(\p v)$.
      Both cases follow from the assumption that $M,N \in \sn{T}$.
      \qed
  \end{itemize}
\end{proof}

\begin{lemma}[Adequacy of prod($\alpha$)]
  \label{lem:sn_bullet}
  If $M \in \sn{T}$, then $\alpha \sdot M \in \sn{T}$.
\end{lemma}
\begin{proof}
  Proof is by induction on the proposition $T$. Note that by Lemma~\ref{lem:dot_sn} we have that $\alpha \sdot M$ always terminates.
  \begin{itemize}
    \item $T$ is $\one$ or $1$:
      Directly follows from Lemma~\ref{lem:dot_sn}.
    \item $T$ is $\p P \pwith \p Q$:
      Since, $\alpha \sdot M$ terminates and $T$ is $\p{P \with Q}$ it follows that $\alpha \sdot M \hookrightarrow^{\ast} \p{\pair{v_1}{v_2}}$.
      Then either $M \hookrightarrow^{\ast} \p{\pair{u_1}{u_2}}$ with $\alpha \sdot \p{u_1} \hookrightarrow^{\ast} \p{v_1}$ and 
      $\alpha \sdot \p{u_2} \hookrightarrow^{\ast} \p{v_2}$ or $M \hookrightarrow^{\ast} \p{\pair{v_1}{v_2}}$. In the first case, the conclusion
      follows by an application of induction hypothesis to $\p{u_1}, \p{u_2}$ and $M \in\sn{T}$. In the second case, the conclusion follows
      from the fact that $M \in \sn{T}$.
    \item $T$ is $\p P \pmultimap \p Q$ or $A \multimap B$:
      We present only one case here, the other case follows similarly.
      Since $\alpha \sdot M$ terminates and $T$ is $\p{P \multimap Q}$, it follows that $\alpha \sdot M \hookrightarrow^{\ast} \p{\lambda \abstr{x}m}$.
      Now either $M \hookrightarrow^{\ast} \p{\lambda \abstr{x}n}$ and $m$ is $\alpha \sdot \p n$
      or $M \hookrightarrow^{\ast} \p{\lambda \abstr{x}m}$. The first case follows from the induction hypothesis applied to $\p{(v/x)n}$ for 
      $\p v \in \sn{\p P}$ and $M \in \sn{T}$. The latter case follows
      from $M \in \sn{T}$.
    \item $T$ is $\p{P \otimes Q}$ or $A \otimes B$:
      Again we present only one case here, the argument for the other case is similar. Since $\alpha \sdot M$ terminates and $T$ is $\p{P \otimes Q}$,
      we have $\alpha \sdot M \hookrightarrow^{\ast} \p{v_1 \otimes v_2}$ and $M \hookrightarrow^{\ast} \p{v_1 \otimes v_2}$. Now the claim follows
      since $M \in \sn{T}$.
    \item $T$ is $A \oplus B$:
      Since $\alpha \sdot M$ terminates and $T$ is $A \oplus B$,
      we have $\alpha \sdot M \hookrightarrow^{\ast} \inl(v)$(or $\inr(v)$) and $M \hookrightarrow^{\ast} \inl(v)$(or $\inr(v)$). Now the claim follows
      since $M \in \sn{T}$.
    \item $T$ is $\mathcal B(\p P)$:
      Since $\alpha \sdot M$ terminates and $T$ is $\mathcal B(\p P)$,
      we have $\alpha \sdot M \hookrightarrow^{\ast} \mathcal B(\p v)$ and $M \hookrightarrow^{\ast} \mathcal B(\p v)$. Now the claim follows
      since $M \in \sn{T}$.
  \end{itemize}
	\qed
\end{proof}

\begin{lemma}[Adequacy of $\zero_e^m$]
\label{lem:sn_elimzero}
        If $m \in \sn{\zero}$ then ${\elimzero(m)} \in \sn{A}$.
\end{lemma}
\begin{proof}
 $m$ always terminates. Any reduction acting on ${\elimzero(m)}$ can only reduce $m$, hence it always terminates. Since, ${\elimzero(m)}$
can never reduce to a closed value, it follows that ${\elimzero(m)} \in \sn{A}$.
\qed
\end{proof}

\begin{lemma}[Adequacy of $\one_e^p$]
  \label{lem:sn_elimone_p}
  If $\p{m_1} \in \sn{\one}$ and $\p{m_2} \in \sn{\p P}$, then $\p{\elimone(m_1,m_2)} \in \sn{\p P}$.
\end{lemma}
\begin{proof}
  We use induction on $\ell(\p{m_1})$ and Lemma~\ref{lem:forward} to prove that $\p{\elimone(m_1,m_2)} \in \sn{\p P}$.
  From Lemma~\ref{lem:forward} it suffices to show that all one-step reducts of the term $\p{\elimone(m_1,m_2)}$
  are in $\sn{\p P}$. Now if a reduction applies to $\p{m_1}$, then by induction hypothesis the resulting
  term is in $\sn{\p P}$. Otherwise, $\p{m_1}$ is of the form $\p a \sdot \p \ast$, in which case,
  $\p{\elimone(m_1,m_2)} \hookrightarrow \p a \sdot \p{m_2}$. Now again from Lemma~\ref{lem:sn_bullet}
  it follows that $\p a \sdot \p{m_2} \in \sn{\p P}$. Hence, from Lemma~\ref{lem:forward} it follows
  that $\p{\elimone(m_1,m_2)} \in \sn{\p P}$.
  \qed
\end{proof}

\begin{lemma}[Adequacy of $1_e^m$]
  \label{lem:sn_elimone_m}
  If $m_1 \in \sn{1}$ and $m_2 \in \sn{A}$, then $\elimone(m_1,m_2) \in \sn{A}$.
\end{lemma}
\begin{proof}
  Proof for this lemma follows by the same argument as for Lemma~\ref{lem:sn_elimone_p}.
  \qed
\end{proof}

\begin{lemma}[Adequacy of $\pmultimap_i^p$]
  \label{lem:sn_arrow_p}
  If $\p v \in \sn{\p P}$ and $\p{(v/x)m} \in \sn{\p Q}$ then $\p{\lambda \abstr{x}m} \in \sn{\p P  \p \multimap \p Q}$.
\end{lemma}
\begin{proof}
  Straightforward from the definition of $\sn{\p{P \multimap Q}}$.
  \qed
\end{proof}

\begin{lemma}[Adequacy of $\multimap_i^m$]
  \label{lem:sn_arrow_m}
  If $v \in \sn{A}$ and $(v/x)m \in \sn{B}$ then $\lambda \abstr{x}m \in \sn{A \multimap B}$.
\end{lemma}
\begin{proof}
  Straightforward from the definition of $\sn{A \multimap B}$.
  \qed
\end{proof}

\begin{lemma}[Adequacy of $\pmultimap_e^p$]
  \label{lem:sn_arrow_ep}
  If $\p n \in \sn{\p P}$ and $\p m \in \sn{\p P \p \multimap \p Q}$ then $\p{m~n} \in \sn{\p Q}$.
\end{lemma}
\begin{proof}
  Proof is by induction on $\ell(m) + \ell(n)$ and Lemma~\ref{lem:forward}. If a reduction applies to either
  $\p m$ or $\p n$, then the resulting term is in $\sn{\p Q}$ by the induction hypothesis. Otherwise, both
  $\p m$ and $\p n$ are values and $\p{m~n}$ reduces to $\p{(n/x)m'}$ for some term $\p{m'}$. Since $\p m \in \sn{\p{P \multimap Q}}$,
  we have that $\p{(n/x)m'} \in \sn{\p Q}$. Hence, all one step reducts of $\p{m~n} \in \sn{\p Q}$. Therefore, by Lemma~\ref{lem:forward}
  it follows that $\p{m~n} \in \sn{\p Q}$.
  \qed
\end{proof}

\begin{lemma}[Adequacy of $\multimap_e^m$]
  \label{lem:sn_arrow_em}
  If $n \in \sn{A}$ and $m \in \sn{A \multimap B}$ then $m~n \in \sn{B}$.
\end{lemma}
\begin{proof}
  Analogous to the proof of Lemma~\ref{lem:sn_arrow_ep}.
  \qed
\end{proof}

\begin{lemma}[Adequacy of $\p\with_i^p$]
  \label{lem:sn_with}
  If $\p{m_1} \in \sn{\p P}$ and $\p{m_2} \in \sn{\p Q}$ then $\p{\pair{m_1}{m_2}} \in \sn{\p P \pwith \p Q}$.
\end{lemma}
\begin{proof}
  Note that since a reduction can only apply to $\p{\pair{m_1}{m_2}}$ if it applies to either $\p{m_1}$ or $\p{m_2}$,
  it follows that since $\p{m_1}$ and $\p{m_2}$ strongly terminate, so does $\p{\pair{m_1}{m_2}}$. Moreover, we also have
  that $\p{\pair{m_1}{m_2}} \hookrightarrow^{\ast} \p{\pair{v_1}{v_2}}$ for $\p{v_1} \in \sn{\p P_1}$ and $\p{v_2} \in \sn{\p P_2}$,
  since $\p{\pair{m_1}{m_2}} \hookrightarrow \p{\pair{m_1'}{m_2'}} \implies \p{m_1} \hookrightarrow \p{m_1'}$ and $\p{m_2} \hookrightarrow \p{m_2'}$.
  \qed
\end{proof}

\begin{lemma}[Adequacy of $\p\with_{e1}^p$]
\label{lem:sn_with_e1}
	If $\p{m_1} \in \sn{\p P \pwith \p Q}$ and for all $\p v \in \sn{\p P}$, $\p{(v/x)m_2} \in \sn{\p R}$,
	then $\p{\elimwith^1(m_1 , \abstr{x}m_2)} \in \sn{\p R}$.
\end{lemma}
\begin{proof}
  Proof is by induction on $\ell(\p{m_1})$ and Lemma~\ref{lem:forward}. If a reduction applies
  to $\p{m_1}$, then we have that the resulting term is in $\sn{\p P \pwith \p Q}$. If 
  $\p{m_1}$ is a value, then we have $\p{(m_1/x)m_2} \in \sn{\p R}$. Hence all one step reducts of 
  $\p{\elimwith^1(m_1 , \abstr{x}m_2)}$ are in $\sn{\p R}$. Therefore, the result follows
  from Lemma~\ref{lem:forward}.
  \qed
\end{proof}

\begin{lemma}[Adequacy of $\p\with_{e1}^p$]
  \label{lem:sn_with_e2}
  If $\p{m_1} \in \sn{\p P \pwith \p Q}$ and for all $\p v \in \sn{\p Q}$, $\p{(v/x)m_2} \in \sn{\p R}$,
  then $\p{\elimwith^2(m_1 , \abstr{x}m_2)} \in \sn{\p R}$.
\end{lemma}
\begin{proof}
  Proof follows by the same argument as that for the proof of Lemma~\ref{lem:sn_with_e1}.
  \qed
\end{proof}

\begin{lemma}[Adequacy of $\oplus_{i1}^m$]
  \label{lem:sn_oplus_l}
  If $m \in \sn{A}$, then $\inl(m) \in \sn{A \oplus B}$.
\end{lemma}
\begin{proof}
  Note that a reduction can only apply to $\inl(m)$ through $m$. Hence, $\inl(m)$ strongly terminates.
  Moreover, since $\inl(m) \hookrightarrow \inl(m') \implies m \hookrightarrow m'$, we have $\inl(m) \in \sn{A \oplus B}$.
  \qed
\end{proof}

\begin{lemma}[Adequacy of $\oplus_{i2}^m$]
  \label{lem:sn_oplus_r}
  If $m \in \sn{B}$, then $\inr(m) \in \sn{A \oplus B}$.
\end{lemma}
\begin{proof}
  This result follows using the same argument as the one for the proof of Lemma~\ref{lem:sn_oplus_l}.
  \qed
\end{proof}

\begin{lemma}[Adequacy of $\oplus_{e}^m$]
  \label{lem:sn_oplus_elim}
  If $m \in \sn{A \oplus B}$ and for all $m_1 \in \sn{A}, m_2 \in \sn{B}$, we have $(m_1/x)n \in \sn{C}$ and $(m_2/y)o \in \sn{C}$,
  then $\elimplus(m, \abstr{x}n, \abstr{y}o) \in \sn{C}$.
\end{lemma}
\begin{proof}
  Proof is by induction on $\ell(m)$ and an application of Lemma~\ref{lem:forward}. If a reduction applies
  to $m$ then the resulting term is in $\sn{A \oplus B}$ as a result of the induction hypothesis. If,
  $m$ is a value, and the resulting term is of the form $(m_1/x)n$ or $(m_2/y)o$. Otherwise, $m$ is of the form
  $\textstyle\sum_{i=1}^n v_i$, where $v_i$ is of the form $\inl(v)$ or $\inr(v)$. Then $\elimplus(m, \abstr{x}n, \abstr{y}o)$
  reduces to $\textstyle\sum_{i=1}^n p_i \sdot \elimplus(v_i, \abstr{x}n, \abstr{y}o)$. This term is in $\sn{A \oplus B}$ by the
  induction hypothesis, Lemma~\ref{lem:sn_bullet} and Lemma~\ref{lem:sn_plus}. Hence, all one-step reducts of
  $\elimplus(m, \abstr{x}n, \abstr{y}o)$ are in $\sn{C}$. Therefore, the claim follows by Lemma~\ref{lem:forward}.
  \qed
\end{proof}

\begin{lemma}[Adequacy of $\p\otimes_i^p$]
  \label{lem:sn_tens_p}
  If $\p{m_1} \in \sn{\p P}$ and $\p{m_2} \in \sn{\p Q}$ then $\p{m_1 \otimes m_2} \in \sn{\p P \p \otimes \p Q}$.
\end{lemma}
\begin{proof}
  Note that since a reduction can only apply to $\p{m_1 \otimes m_2}$ if it applies to either $\p{m_1}$ or $\p{m_2}$,
  it follows that since $\p{m_1}$ and $\p{m_2}$ strongly terminate, so does $\p{m_1 \otimes m_2}$. Moreover, we also have
  that $\p{m_1 \otimes m_2} \hookrightarrow^{\ast} \p{v_1 \otimes v_2}$ for $\p{v_1} \in \sn{\p P_1}$ and $\p{v_2} \in \sn{\p P_2}$,
  since $\p{m_1 \otimes m_2} \hookrightarrow \p{m_1' \otimes m_2'} \implies \p{m_1} \hookrightarrow \p{m_1'}$ and $\p{m_2} \hookrightarrow \p{m_2'}$.
  \qed
\end{proof}

\begin{lemma}[Adequacy of $\otimes_i^m$]
  \label{lem:sn_tens_m}
  If $m \in \sn{A}$ and $n \in \sn{B}$ then $m \otimes n \in \sn{A \otimes B}$.
\end{lemma}
\begin{proof}
  The result follows using the same argument as the one for Lemma~\ref{lem:sn_tens_m}.
  \qed
\end{proof}

\begin{lemma}[Adequacy of $\p\otimes_{e}^p$]
  \label{lem:sn_tenselim_p}
  If $\p{m_1} \in \sn{\p P \p \otimes \p Q}$ and for all $\p{v_1} \in \sn{\p P}, \p{v_2} \in \sn{\p Q}$
  we have $\p{(v_1/x,v_2/y)m_2} \in \sn{\p R}$ then $\p{\elimtens(m_1, \abstr{xy} m_2)}$ $\in \sn{\p R}$.
\end{lemma}
\begin{proof}
  Again proof is by induction on $\ell(\p{m_1})$ and an application of Lemma~\ref{lem:forward}. If a reduction applies
  to $\p{m_1}$ then the resulting term is in $\sn{\p{P \otimes Q}}$ as a result of the induction hypothesis. If,
  $\p{m_1}$ is a value, and the resulting term is of the form $\p{(v_1/x,v_2/y)m_2}$. Otherwise $\p{m_1}$ is of the 
  form $\textstyle\sum_{i=1}^n \p{a_i} \sdot \p{v_n}^1 \p \otimes \p{v_n}^2$. In this case $\p{\elimtens(m_1, \abstr{xy} m_2)}$ reduces to 
  $\textstyle\sum_{i=1}^n \p{a_i} \sdot \p{\elimtens(} \p{v_n}^1 \p \otimes \p{v_n}^2 \p{, \abstr{xy}m_2)}$. By induction hypothesis, Lemma~\ref{lem:sn_bullet} and
  Lemma~\ref{lem:sn_plus}, this term is in $\sn{\p R}$. Hence, all one-step reducts of
  $\p{\elimtens(m_1, \abstr{xy} m_2)}$ are in $\sn{\p R}$. Therefore, the claim follows by Lemma~\ref{lem:forward}.
  \qed
\end{proof}

\begin{lemma}[Adequacy of $\otimes_{e}^m$]
  \label{lem:sn_tenselim_m}
  If $m_1 \in \sn{A \otimes B}$ and for all $m \in \sn{A}, n \in \sn{B}$
  $(m/x,n/y)m_2 \in \sn{C}$ then $\elimtens(m_1, \abstr{xy} m_2) \in \sn{C}$.
\end{lemma}
\begin{proof}
  The proof for this lemma follows using the same argument as the one for Lemma~\ref{lem:sn_tenselim_p}.
  \qed
\end{proof}

\begin{lemma}[Adequacy of $\mathcal B^m$]
  \label{lem:sn_l}
  If $\p m \in \sn{\p P}$, then $\mathcal B(\p m) \in \sn{\mathcal B(\p P)}$.
\end{lemma}
\begin{proof}
  Note that any reduction that applies to $\mathcal B(\p m)$ is through $\p m$. Hence, $\mathcal B(\p m)$
  strongly terminates. Since, $\mathcal B(\p m) \hookrightarrow \mathcal B(\p m') \implies \p m \hookrightarrow \p m'$,
  we have $\mathcal B(\p m) \in \sn{\mathcal B(\p P)}$.
  \qed
\end{proof}

\begin{lemma}[Adequacy of $\mathcal B(\p\multimap)_e^m$]
  \label{lem:sn_arrow_l}
  If $m \in \sn{\mathcal B(\p{P \multimap Q})}$ and $n \in \mathcal B(\p P)$ then $m~n \in \sn{\mathcal B(\p Q)}$.
\end{lemma}
\begin{proof}
  Again, proof is by induction on $\ell(m) + \ell(n)$ and Lemma~\ref{lem:forward}. We prove all one-step reducts of $m~n$ are in
  $\sn{\mathcal B(\p Q)}$. If a reduction applies to either $m$
  or $n$, then the resulting term is in $\sn{\mathcal B(\p Q)}$ by induction hypothesis. Otherwise $m$ is of the form $\textstyle\sum_{i=1}^n p_i \sdot \mathcal B(\p{\lambda \abstr{x}m_i'})$
  and $n$ is of the form $\textstyle\sum_{j=1}^k q_j \sdot \mathcal B(\p{v_j})$. In this case, $m~n$ reduces to $\textstyle\sum_{i=1}^n\textstyle\sum_{j=1}^kp_iq_j \sdot \mathcal B(\p{\lambda \abstr{x}m'~n})$.
  Now the claim follows
  by Lemma~\ref{lem:sn_arrow_ep}, Lemma~\ref{lem:sn_bullet}, Lemma~\ref{lem:sn_plus} and the definition of $\sn{\mathcal B(\p P)}$ for any proposition $\p P$.
  \qed
\end{proof}

\begin{lemma}[Adequacy of $\mathcal B(\otimes)^m$]
  \label{lem:sn_tau}
  If $m \in \sn{\mathcal B(\p P) \otimes \mathcal B(\p Q)}$, then $\tau(m) \in \sn{\mathcal B(\p{P \otimes Q})}$.
\end{lemma}
\begin{proof}
  Proof is by  induction on $\ell(\p m)$ and Lemma~\ref{lem:forward}. If a reduction applies to $m$ then the resulting term is in 
  $\sn{\mathcal B(\p{P \otimes Q})}$ by induction hypothesis. If $m$ is a value of the form $\mathcal B(\p{v_1}) \otimes \mathcal B(\p{v_2})$
  then the term reduces to $\mathcal B(\p{v_1 \otimes v_2})$ which is in $\sn{\mathcal B(\p{P \otimes Q})}$, from Lemma~\ref{lem:sn_tens_p}
  and $m \in \sn{\mathcal B(\p P) \otimes \mathcal B(\p Q)}$. Otherwise $m$ is of the form $\textstyle\sum_{i=1}^n p_i \sdot v_i$, in which case one of the
  following rules may apply:
  \begin{align*}
    \tau(m) &\hookrightarrow \textstyle\sum_{i=1}^n p_i \sdot \tau(v_i)\\
    \tau(m) &\hookrightarrow \tau(p_i \sdot v_i)
  \end{align*}
  The resulting term in both the cases is in $\sn{\mathcal B(\p{P \otimes Q})}$ by induction hypothesis,
  Lemma~\ref{lem:sn_bullet} and Lemma~\ref{lem:sn_plus}. Hence, all one step reducts of $\tau(m)$ are in $ \sn{\mathcal B(\p{P \otimes Q})}$.
  Now the result follows by Lemma~\ref{lem:forward}.
  \qed
\end{proof}

\begin{lemma}[Adequacy]
\label{lem:subst-strong-norm}
For the judgement $\Gamma \vdash M : T$, if $\theta \vDash \Gamma$ then $\theta(M) \in \sn T$.
\end{lemma}
\begin{proof}
  Proof is by induction on the derivation of $\Gamma\vdash M:T$.
  \begin{itemize}
    \item If the last rule is $ax$, $M=x$: Since $\theta\vDash\Gamma$, we have $\theta(x)\in\sn T$.
    \item If the last rule is sum, $M=M' \Plus N'$: Since $\theta(M' \Plus N') = \theta M' \Plus \theta N'$, from the induction
      hypothesis and Lemma~\ref{lem:sn_plus} it follows that $\theta M \in \sn T$.
    \item If the last rule is prod$(\alpha)$: $M$ is $\alpha \sdot N$: Since $\theta (\alpha \sdot N) = \alpha \sdot (\theta N)$, from the induction
      hypothesis and Lemma~\ref{lem:sn_bullet} it follows that $\theta M \in \textbf{SN}(T)$.
    \item If the last rule is $\topintro_i^p$, $M$ is $\p\topintro$: Since $M$ is closed and irreducible we have $\theta M = M$, and $M \in \sn{\p\top}$.
    \item If the last rule is $0_e^m$, $M$ is $\elimzero(m)$: Note that $\theta M = \elimzero(\theta m)$. Now the claim follows from the induction hypothesis 
      applied to $\theta  m $ and Lemma~\ref{lem:sn_elimzero}.
    \item If the last rule is $1_i$, $M$ is $\ast$: Then $M$ is closed irreducible proof term of $1$. Hence, $\theta M = M$, and we have $\theta M \in \sn{T}$.
    \item If the last rule is $1_e$, $M$ is $\elimone(m,n)$: Note that $\theta {\elimone(m,n)} = {\elimone(}\theta {m, } \theta {n)}$.
      Now the claim follows from the induction hypothesis and Lemmas~\ref{lem:sn_elimone_p} and \ref{lem:sn_elimone_m}.
    \item If the last rule is $\multimap_i$, $M$ is $\lambda \abstr{x}N$: 
      Note that $\theta ({\lambda \abstr{x}N}) = {\lambda \abstr{x}(}\theta {N)}$. 
      We need to show that $\theta({(V/x)N}) \in \sn{ T}$ for any $ V \in \sn{ S}$. However, note that $\theta({(V/x)N}) = \theta'( N)$ for a new substitution function,
      $\theta' = \theta \cup \{x \mapsto  V \}$. Since $ V \in \sn{ S}$, we have $\theta' \vDash {\Gamma, x : S}$.
      Hence, by the induction hypothesis, it follows that $\theta'( M)  \in \sn{ T}$. Hence, the claim follows from
      Lemmas~\ref{lem:sn_arrow_p} and~\ref{lem:sn_arrow_m}.
    \item If the last rule is $\multimap_e$, $M$ is $M'~N'$: Note that $\theta M = (\theta M')~(\theta N')$. Now the claim follows using the induction hypothesis on the terms
      $\theta M', \theta N'$ and applying Lemma~\ref{lem:sn_arrow_ep}, Lemma~\ref{lem:sn_arrow_em}.
    \item If the last rule is $\p\with_i^p$, $M$ is $\p{\pair{m_1}{m_2}}$: Note that $\theta M = \p{\langle} \theta \p{m_1 ,} \theta \p{m_2 \rangle}$. Now the claim follows from the
          induction hypothesis and Lemma~\ref{lem:sn_with}.
	\item If the last rule is $\p\with_{ei}^p$, $M$ is $\p{\elimwith^i(m_1, \abstr{x} m_2)}$: Thus, $\theta M = \p{\elimwith^i(}\theta \p{m_1, \abstr{x}}\theta \p{m_2)}$. Now 
	  the claim follows from the induction hypothesis and Lemma~\ref{lem:sn_with_e1}, Lemma~\ref{lem:sn_with_e2}.
    \item If the last rule is $\oplus_{i1}^m$, $M$ is $\inl(m)$: In this case $\theta M = \inl(\theta m)$. Now the claim follows from the induction hypothesis and Lemma~\ref{lem:sn_oplus_l}.
    \item If the last rule is $\oplus_{i2}^m$, $M$ is $ \inr(m)$: In this case $\theta M = \inr(\theta m)$. Now the claim follows from the induction hypothesis and Lemma~\ref{lem:sn_oplus_r}.
    \item If the last rule is $\oplus_e^m$, $M$ is $\elimplus(m, \abstr{x} n, \abstr{y} o)$: Thus, $\theta M = \elimplus(\theta m, \abstr{x} \theta n, \abstr{y} \theta o)$. Now the 
	  claim follows from the induction hypothesis and Lemma~\ref{lem:sn_oplus_elim}.
    \item If the last rule is $\otimes_i$, $M$ is  $\p{m \otimes n}$ or $m \otimes n$: We prove one case, the other follows analogously. 
          Note that $\theta M = \theta \p{m} \p \otimes \theta \p{n}$. Hence from the induction hypothesis
	  it follows that $\theta \p{m} \in \sn{\p P}$ and $\theta \p{n} \in \sn{\p Q}$. Now the claim follows from Lemmas~\ref{lem:sn_tens_p} and~\ref{lem:sn_tens_m}. 
    \item If the last rule is $\otimes_e$, $M$ is $\p{\elimtens(m , \abstr{xy}n)}$ or $\elimtens(m , \abstr{xy}n)$: Again we prove one case, the other follows analogously.
      In this case $\theta M = \p{\elimtens(}\theta \p{m, \abstr{xy}}\theta \p{n)}$. Hence the claim follows from the induction hypothesis and
      Lemmas~\ref{lem:sn_tenselim_p} and~\ref{lem:sn_tenselim_m}. 
    \item If the last rule is $\mathcal B^m$, $M$ is $\mathcal B(\p m)$: Since $\p m$ is a closed term $\theta \p{m} = \p m$, hence from the induction hypothesis and Lemma~\ref{lem:sn_l},
      $\theta M = \mathcal B(\p m) \in \sn{\p P}$.
    \item If the last rule is $\mathcal{B}(\p\multimap)^m_e$, $M$ is $m~n$: Note that $\theta M = (\theta m)~(\theta n)$. Now the claim follows using the induction hypothesis on the terms
	  $\theta m, \theta n$ and applying Lemma~\ref{lem:sn_arrow_l}.
    \item If the last rule is $\mathcal{B}(\p\otimes)^m$, $M$ is $\tau(m)$: Claim follows by the observation that $\theta \tau(m) = \tau(\theta m)$, an application of 
	  induction hypothesis and Lemma~\ref{lem:sn_tau}.
    \item If the last rule is $ex$, then the claim directly follows from the induction hypothesis.
  \end{itemize}
  \qed
\end{proof}

\strongnorm*
\begin{proof}
  The result follows from Lemma~\ref{lem:subst-strong-norm}, by using the identity function as the substitution function.
  \qed
\end{proof}

\section{Omitted Proofs in Section~\ref{sec:denotational}}\label{app:denotational}
\begin{lemma}[Substitution]
\label{lem:subst_d}
Let $\Gamma, x : S \vdash M : T$, $\Delta \vdash N : S$, and
$\Gamma , \Delta \vdash (N/x)M : T$ be valid judgements.  
Then:
\[
  \lrb{\Gamma , \Delta \vdash (N/x)M : T}
  = \lrb{\Gamma, x : S \vdash M : T}
  \circ (id_{\lrb{\Gamma}} \otimes \lrb{\Delta \vdash N : S}).
\]
\end{lemma}
\begin{proof}
  By induction on the derivation of $\Gamma,x:S\vdash M:T$.
  We consider the last rule in the derivation.
  \begin{description}
    \item[({\normalfont ax})] $x : S \vdash x : S$: In this case $(N/x)M=N$, $T=S$, and $\Gamma = \emptyset$.
    \item[({\normalfont sum})] $\Gamma , x : S \vdash M_1 \Plus M_2 : T$: This case follows directly from the induction hypothesis and the observation that $(N/x)(M_1 \Plus M_2) = (N/x)M_1 \Plus (N/x)M_2$.
    \item[({\normalfont prod}$(\alpha)$)] $\Gamma , x : S \vdash \alpha \sdot M : S$: This case also follows directly from the induction hypothesis and the observation that 
      $(N/x)\alpha \sdot M = \alpha \sdot (N/x)M$.
    \item[($\top_i^p$)] $\p{\Gamma \vdash \topintro : \top}$: This case follows from $!_{\lrb{\Gamma}\otimes\lrb{\Delta}} = {!}_{\lrb{\Gamma}\otimes\lrb{S}}\circ(id\otimes\lrb{\Delta \vdash N : \top})$.
    \item[($0_e^m$)] $\Gamma , \Delta \vdash \elimzero(m) : A$: In this case, first observe that $(n/x)\elimzero(m) = \elimzero((n/x)m)$. Now the claim follows from induction hypothesis applied to $m$
      and the following computation:
      \begin{align*}
	&\lrb{\Gamma, \Delta \vdash \elimzero((n/x)m) : A}\\
	&= !_{\zero\otimes\lrb{\Gamma}, A} \circ (\lrb{(n/x)m} \otimes id)\\
	&= !_{\zero\otimes\lrb{\Gamma}, A} \circ (\lrb{m} \otimes id) \circ (\lrb{n} \otimes id)
      \end{align*}
    \item[($1_i$)] $\vdash \ast : \one$: This case is impossible, since $\Gamma,x:S\neq\emptyset$.
    \item[($1_e$)] $\Gamma , x : S \vdash \elimone(M_1,M_2) : T$: In this case we have $(N/x)\elimone(M_1,M_2)=$ $\elimone((N/x)M_1, (N/x)M_2)$.
      We consider the case when $x$ appears in $M_1$, the other case follows using a similar   
      argument. Let $\Gamma = \Gamma_1, \Gamma_2$ so that $\Gamma_1 , x : S \vdash M_1 : \one$ and
      $\Gamma_2 \vdash M_2 : T$. Then we have,
      $\lrb{\Gamma, \Delta \vdash \elimone((N/x)M_1 , M_2) : T}$ $= \lambda_T \circ (\lrb{(N/x)M_1} \otimes \lrb{M_2})$.
      Applying the induction hypothesis, we have that $\lrb{\Gamma_1, \Delta \vdash (N/x)M_1 : \one}$
      $= \lrb{\Gamma_1, x : S \vdash M_1 : \one} \circ (id_{\lrb{\Gamma_1}} \otimes \lrb{\Delta \vdash N : \one})$.
      Hence we have,
      \begin{align*}
	&\lrb{\Gamma, \Delta \vdash \elimone((N/x)M_1, M_2) : T} \\
	&=\lambda_T \circ ((\lrb{\Gamma_1, x : S \vdash M_1 : \one}
	\circ (id_{\lrb{\Gamma_1}} \otimes \lrb{\Delta \vdash N : S})) \otimes \lrb{M_2}) \\
	&=\lambda_T \circ ((\lrb{\Gamma_1, x : S \vdash M_1 : \one} \otimes \lrb{M_2})
	\circ (id_{\lrb{\Gamma_1}} \otimes \lrb{\Delta \vdash N : S} \otimes id_{\lrb{\Gamma_2}})) \\
	&=\lrb{\Gamma, x : S \vdash \elimone(M_1,M_2) : T} \circ (id_{\lrb{\Gamma}} \otimes \lrb{\Delta \vdash N : S})
      \end{align*}
    \item[($\multimap_i$)] $\Gamma , x : S \vdash \lambda \abstr{y}M : {T \multimap R}$: We have $(N/x)\lambda \abstr{y}M = \lambda \abstr{y}(N/x)M$. We use the induction hypothesis
      on $(N/x)M$ to obtain the following:
      \begin{align*}
	&\lrb{\Gamma , \Delta \vdash \lambda \abstr{y}(N/x)M : {T \multimap R}} \\      
	&=\Phi (\lrb{\Gamma , \Delta , y : T \vdash (N/x)M : R}) \\
	&=\Phi (\lrb{\Gamma , y : T , x : S \vdash M : R}
	\circ (id_{\lrb{\Gamma, y : T}} \otimes \lrb{\p{\Delta \vdash N : S}})) \\
	&=\lrb{\Gamma, x : S \vdash \lambda \abstr{y}m : {T \multimap R}} \circ (id_{\lrb{\Gamma}} \otimes \lrb{\Delta \vdash n : S})
      \end{align*}
    \item[($\multimap_e$)] $\Gamma , x : S \vdash M_1~M_2 : T$:  We have $(N/x)(M_1~M_2) = (N/x)M_1~(N/x)M_2$. We consider the case when $x$ appears in $M_1$, the other
      case follows using a similar argument. Again we let $\Gamma = \Gamma_1 , \Gamma_2$ so that $\Gamma_1 , x : S \vdash M_1 : {R \multimap T}$, and
      $\Gamma_2 \vdash M_2 : R$. Let $f,g$ denote $\lrb{\Gamma_1 , x : S \vdash M_1 : {R \multimap T}},$ and $\lrb{\Gamma_2 \vdash M_2 : R}$ respectively.
      \begin{align*}
	&\lrb{\Gamma, \Delta \vdash (N/x)M_1~M_2 : T} \\
	&=\textit{eval}_{T, R}
	\circ (\lrb{\Gamma_1, \Delta \vdash (N/x)M_1 : {R \multimap T}} \otimes \lrb{\Gamma_2 \vdash M_2 : R}) \\
	&=\textit{eval}_{T, R} \circ ((f \circ (id_{\lrb{\Gamma_1}} \otimes \lrb{\Delta \vdash N : S})) \otimes g) \\
	&=\textit{eval}_{T, R} \circ (f \otimes g) \circ (id_{\lrb{\Gamma_1},\Gamma_2} \otimes \lrb{\Delta \vdash N : S})\\
	&=\lrb{\Gamma , x : S \vdash M_1~M_2 : T} \circ (id_{\Gamma} \otimes \lrb{\Delta \vdash N : S})
      \end{align*}
    \item[($\pwith_i^p$)] $\p{\Gamma , x : R \vdash \pair{m_1}{m_2} : P \with Q}$: Note that we have $\p{(n/x)\pair{m_1}{m_2}}=$ $\p{\pair{(n/x)m_1}{(n/x)m_2}}$.
      We use the induction hypothesis for both the terms $\p{(n/x)m_1} , \p{(n/x)m_2}$. We have:
      \begin{align*}
	&\lrb{\p{\Gamma \vdash (n/x)\pair{m_1}{m_2} : P \with Q}}\\
	&=\langle \lrb{\p{\Gamma \vdash (n/x)m_1 : P}} , \lrb{\p{\Gamma \vdash (n/x)m_2 : Q}}\rangle\\
	&= \langle \lrb{\p{\Gamma, x : R \vdash m_1 : P}}\circ (id_{\p \Gamma} \otimes \lrb{\p{\Delta \vdash n : R}}),\\
	&\hspace{2cm}\lrb{\p{\Gamma, x : R \vdash m_2 : Q}}\circ (id_{\p\Gamma} \otimes \lrb{\p{\Delta \vdash n : R}})\rangle\\
	&= \langle \lrb{\p{\Gamma, x : R \vdash m_1 : P}} , \lrb{\p{\Gamma, x : R \vdash m_2 : Q}} \rangle
	\circ (id_{\p \Gamma} \otimes \lrb{\p{\Delta \vdash n : R}})\\
	&= \lrb{\p{\Gamma , x : R \vdash \pair{m_1}{m_2} : P \with Q}} \circ (id_{\p\Gamma} \otimes \lrb{\p{\Delta \vdash n : R}})
      \end{align*}
    \item[($\pwith_{ei}^p$)] $\p{\Gamma , x : P' \vdash \elimwith^i(m', \abstr{y}n') : P}$: We present the case where $\p x$ appears in $\p m'$ first. Let $\p \Gamma = \p{\Gamma_1}, \p{\Gamma_2}$,
      so that $\p{\Gamma_1 , x : P' \vdash m' : Q \with R}$ and $\p{\Gamma_2 , y : Q \vdash n' : P}$ are valid. We have:
      \begin{align*}
	&\p{\lrb{\Gamma ,\Delta \vdash \elimwith^i((n/x)m', \abstr{y}n') : P}}\\
	&= \lrb{\p n'} \circ (\pi_i \otimes id_{\p \Gamma_2}) \circ (\lrb{\p{(n/x)m'}} \otimes id_{\p \Gamma_2})
	\circ  (id_{\lrb{\p \Gamma_1}} \otimes \sigma_{\lrb{\p \Gamma_2}, \lrb{\p \Delta}})\\
	&= \lrb{\p n'} \circ (\pi_i \otimes id_{\p \Gamma_2}) \circ (\lrb{\p m'} \otimes id_{\p \Gamma_2})
	\begin{aligned}[t]
	  &\circ (id_{\p \Gamma_1} \otimes \lrb{\p n} \otimes id_{\p \Gamma_2})\\
	  &\circ (id_{\lrb{\p \Gamma_1}} \otimes \sigma_{\lrb{\p \Gamma_2}, \lrb{\p \Delta}})
	\end{aligned}\\
	&= \lrb{\p n'} \circ (\pi_i \otimes id_{\p \Gamma_2}) \circ (\lrb{\p m'} \otimes id_{\p \Gamma_2}) \circ(id_{\p \Gamma} \otimes \lrb{\p n})\\
	&= \lrb{\p{\Gamma, x : P' \vdash \elimwith^i(m', \abstr{y}n') : P}} \circ (id_{\p \Gamma} \otimes \lrb{\p n})
      \end{align*}
      For the case when $\p x$ appears in $\p n'$, again let $\p \Gamma = \p \Gamma_1 , \p \Gamma_2$ so that
      $\p{\Gamma_1 \vdash m' : Q \with R}$, $\p{\Gamma_2 , \Delta \vdash (n/x)n' : P}$ are valid.
      We have $\p{\Delta \vdash n : P'}$. It suffices to show that $((\pi_1 \circ \lrb{\p n}) \otimes id_{\lrb{\p{\Gamma_2, P'}}}) \circ (id_{\lrb{\p \Gamma}} \otimes \lrb{\p n}) =$
      $(id_{\lrb{\p {Q, \Gamma_2}}} \otimes \lrb{\p n}) \circ ((\pi_1 \circ \lrb{\p m'}) \otimes id_{\lrb{\p{\Gamma_2}, \p \Delta}})$. But this holds,
      because $\p m',\p n$ are independent of each other, and act on disjoint contexts.
    \item[($\oplus_{i1}^m$)] $\Gamma , x : C \vdash \inl(m') : A \oplus B$: We have $(n/x)\inl(m') = \inl((n/x)m')$. We apply the induction hypothesis on $(n/x)m'$. Hence,
      \begin{align*}
	&\lrb{\Gamma , \Delta \vdash (n/x)\inl(m') : A \oplus B}\\
	&= i_1 \circ (\lrb{\Gamma, \Delta  \vdash (n/x)m' : A})\\
	&= i_1 \circ \lrb{\Gamma, x: C \vdash m' : A} \circ (id_{\Gamma} \otimes \lrb{\Delta \vdash n : C})\\
	&= \lrb{\Gamma, x: C \vdash \inl(m') : A \oplus B} \circ (id_{\Gamma} \otimes \lrb{\Delta \vdash n : C})
      \end{align*}
    \item[($\oplus_{i2}^m$)] $\Gamma , x : C \vdash \inr(m') : A \oplus B$: This case follows using the same argument as that for the above case.
    \item[($\oplus_{e}^m$)] $\Gamma , x : A' \vdash \elimplus(m' , \abstr{y}n', \abstr{z}o') : C$: Suppose $x$ appears in $m'$. In this case $(n/x)\elimplus(m' , \abstr{y}n', \abstr{z}o') =$
      $\elimplus((n/x)m', \abstr{y}n' , \abstr{z}o')$. Let $\Gamma = \Gamma_1, \Gamma_2$ so that
      $\Gamma_1 , x : A' \vdash m' : A \oplus B$, $\Gamma_2, y: A  \vdash n' : C$ and $\Gamma_2, z : B \vdash o' : C$ are valid. Let $f,g,h$ denote $\lrb{\Gamma_2 ,y : A \vdash n' : C},$
      $\lrb{\Gamma_2, z : B \vdash o' : C},$ $\lrb{\Gamma_1, x : A' \vdash m' : A\oplus B}$ respectively. We have:
      \begin{align*}
	&\lrb{\Gamma \vdash \elimplus((n/x)m', \abstr{y}n' , \abstr{z}o' : C)}\\
	&= [f,g] \circ d \circ \sigma \circ ((h \circ (id \otimes \lrb{n}))\otimes id)\\
	&= [f,g] \circ d \circ \sigma \circ (h \otimes id) \circ (id \otimes \lrb{n} \otimes id)\\
	&= [f,g] \circ d \circ \sigma \circ (h \otimes id) \circ (id \otimes \lrb{n})
      \end{align*}
      Suppose $x$ appears in $n',o'$. Then the term $(n/x)\elimplus(m', \abstr{y} n', \abstr{z} o')$ is equal to $\elimplus(m', \abstr{y}(n/x)n', \abstr{z}(n/x)o')$.
      Again let $\Gamma = \Gamma_1 , \Gamma_2$ so that $\Gamma_1 \vdash m : A \oplus B$, $\Gamma_2 , x : A' , y : A \vdash n : C$ and
      $\Gamma_2, x : A', z : B \vdash o :C$ are valid. Let $f,g,h$ denote $\lrb{\Gamma_2, y : A, x : A' \vdash n' : C}$, $\lrb{\Gamma_2 , z : B, x : A' \vdash o' : C}$,
      $\Gamma_1 \vdash m' : A \oplus B$ respectively. We use the induction hypothesis on $(n/x)n', (n/x)o'$. We have:
      \begin{align*}
	&\lrb{\Gamma_1,\Gamma_2,\Delta \vdash \elimplus(m',\abstr{y}(n/x)n', \abstr{z}(n/x)o') :C}\\
	&= [ f\circ (id \otimes \lrb{n}), g \circ(id \otimes \lrb{n}) ]
	\circ d \circ \sigma \circ (h \otimes id)\\
	&= [f,g] \circ (id \otimes \lrb{n})
	\circ d \circ \sigma \circ (h \otimes id)\\
	&= [f,g] \circ d \circ \sigma \circ (h \otimes id)
	\circ id \otimes \lrb{n}
      \end{align*}
      Note that the last equality holds because $m',n$ are independent terms and acting on disjoint contexts, therefore the corresponding morphisms commute.

    \item[($\otimes_i$)] $\Gamma , x : S \vdash M_1 \otimes M_2 : {T_1 \otimes T_2}$: Without loss of generality assume $x$ appears in $M$. The case when $x$ appears in $N$ follows using
      a similar argument. In this case we have, $(N/x)(M_1 \otimes M_2) = (N/x)M_1 \otimes M_2$. Let $\Gamma = \Gamma_1, \Gamma_2$ so that $\Gamma_1, x : S \vdash M_1 : {T_1}$
      and $\Gamma_2 \vdash M_2 : {T_2}$ are valid. We use the induction hypothesis for $(N/x)M_1$. Let $f,f',g$ denote $\lrb{\Gamma_1, \Delta \vdash (N/x)M_1 : {T_1}}$,
      $\lrb{\Gamma_1, x : S \vdash M_1 : {T_1}}$, $\lrb{\Gamma_2 \vdash M_2 : {T_2}}$ respectively.
      \begin{align*}
	&\lrb{\Gamma_1, \Gamma_2, \Delta \vdash (N/x)M_1 \otimes M_2 : {T_1\otimes T_2}}\\
	&=(f \otimes g) \circ (id_{\Gamma_1} \otimes \sigma_{\lrb{\Gamma_2}, \lrb{\Delta}})\\
	&= ((f' \circ (id_{\Gamma_1} \otimes \lrb{N})) \otimes g) \circ (id_{\Gamma_1} \otimes \sigma_{\lrb{\Gamma_2}, \lrb{\Delta}})\\
	&= (f' \otimes g) \circ (id_{\Gamma_1} \otimes \lrb{N}\otimes id_{\Gamma_2}) \circ (id_{\Gamma_1} \otimes \sigma_{\lrb{\Gamma_2}, \lrb{\Delta}})\\
	&= (f' \otimes g) \circ (id_{\Gamma} \otimes \lrb{N})
      \end{align*}
    \item[($\otimes_e$)] $\Gamma , x : S \vdash \elimtens(M' , \abstr{y z}N') : T$:
      Again there are two cases here. First we consider the case when $x$
	    appears in $M'$. Applying the induction hypothesis on $M'$, we have the following:
	    \begin{align*}
		    &\lrb{\Gamma , \Delta \vdash \elimtens((N/x)M', \abstr{y z}N') : T}\\
		    &= \lrb{N'} \circ (id \otimes \lrb{(N/x)M'})\\
		    &=\lrb{N'} \circ (id \otimes(\lrb{M'} \circ (id \otimes \lrb{N})))\\
		    &=\lrb{N'} \circ (id \otimes \lrb{M'}) \circ (id \otimes \lrb{N})\\
		    &=\lrb{\elimtens(M', \abstr{y z}N')} \circ (id \otimes \lrb{N})
	    \end{align*}
	    Now we consider the case when $x$ appears in $N'$. Applying the induction hypothesis on $N'$, we have the following:
	    \begin{align*}
		    &\lrb{\Gamma , \Delta \vdash \elimtens((N/x)M', \abstr{y z}N') : T}\\
		    &= \lrb{(N/x)N'} \circ (id \otimes \lrb{M'})\\
		    &= \lrb{N'} \circ (id \otimes \lrb{N}) \circ (id \otimes \lrb{M'})\\
		    &= \lrb{N'} \circ (id \otimes \lrb{M'}) \circ (id \otimes \lrb{N})
	    \end{align*}
	    Note that the last equality holds because $M',N$ are independent terms and acting on disjoint contexts, therefore the corresponding morphisms commute.
    \item[($\mathcal B^m$)] $\vdash\mathcal B(\p m) : \mathcal B(\p P)$: This case is impossible, since $\Gamma,x:S\neq\emptyset$.
    \item[($\mathcal B(\multimap)^m$)] $\Gamma , x : A \vdash m'~n' :\mathcal B(\p Q)$: Again we consider two cases, one when $x$ appears in $m'$ and the other when it appears in $n'$.
	  First we consider the case when it appears in $m'$. Applying the induction hypothesis to $m'$, we have the following:
	  \begin{align*}
	     &\lrb{\Gamma, \Delta \vdash (n/x)m'~n' :\mathcal B(\p Q)}\\
	     &= \mathcal B(eval) \circ \tau \circ (\lrb{(n/x)m'} \otimes \lrb{n'})\\
	     &= \mathcal B(eval) \circ \tau \circ ((\lrb{m'}\circ(id \otimes \lrb{n})) \otimes \lrb{n'})\\
	     &= \mathcal B(eval) \circ \tau \circ (\lrb{m'} \otimes \lrb{n'}) \circ (id \otimes \lrb{n})
	  \end{align*}
	  Similarly, for the case when $x$ appears in $n'$, we have the following:
	  \begin{align*}
                  &\lrb{\Gamma, \Delta \vdash (n/x)m'~n' :\mathcal B(\p Q)}\\
		  &= \mathcal B(eval) \circ \tau \circ (\lrb{m'} \otimes \lrb{(n/x)n'})\\
		  &= \mathcal B(eval) \circ \tau \circ (\lrb{m'} \otimes (\lrb{n'}\circ (id \otimes \lrb{n})))\\
                  &= \mathcal B(eval) \circ \tau \circ (\lrb{m'} \otimes \lrb{n'}) \circ (id \otimes \lrb{n})
	  \end{align*}
    \item[($\mathcal B(\otimes)^m$)] $\Gamma , x : A\vdash \tau(m') : B$: This case follows easily from the observation that $(n/x)\tau(m') = \tau((n/x)m')$ and the application of the induction hypothesis
      on $m'$.
    \item[({\normalfont ex})] $\Gamma,x:S_2,y:S_1\vdash M:T$. Direct from the induction hypothesis.
      \qed
  \end{description}
\end{proof}

\begin{restatable}{lemma}{psubsteq}
  \label{lem:subst_eq_p}
  \label{lem:subst_eq_m}
  The following equations hold for valid judgements:
  \begin{enumerate}
    \item$\lrb{\Gamma \vdash \elimone( \alpha\sdot \ast , M) : T} = \lrb{\Gamma \vdash \alpha \sdot M : T}$
    \item$\lrb{\Gamma \vdash (\lambda \abstr{x}M)~N : T} = \lrb{\Gamma \vdash (N/x)M : T}$
    \item$\lrb{\p{\Gamma \vdash \elimwith^i(\pair{m_1}{m_2} , \abstr{x}n) : P}} = \lrb{\p{\Gamma \vdash (m_i/x)n : P}}$
    \item $\lrb{\Gamma \vdash \elimplus(\inl(m), x.n, y.o) : A} = \lrb{\Gamma \vdash (m/x)n : A}$
    \item $\lrb{\Gamma \vdash \elimplus(\inr(m), x.n, y.o) : A} = \lrb{\Gamma \vdash (m/y)o : A}$
    \item$\lrb{\Gamma \vdash \elimtens(M_1 \otimes M_2, \abstr{xy}N) : T} = \lrb{\Gamma \vdash (M_1/x, M_2/y)N : T}$
  \end{enumerate}
\end{restatable}
\begin{proof}
  Most cases follow from Lemma~\ref{lem:subst_d}:
  \begin{enumerate}
    \item 
      $\begin{aligned}[t]
	\lrb{\Gamma \vdash \elimone(\alpha\sdot \ast, M) : T}
	&=\lambda_{T} \circ (\alpha \sdot id_I \otimes \lrb{\Gamma \vdash M : T})\\
	&=\lambda_{T} \circ (id_I \otimes \alpha \sdot \lrb{\Gamma \vdash M : T})\\
	&=\alpha \sdot \lrb{\Gamma \vdash M : T}\\
	&=\lrb{\Gamma \vdash \alpha \sdot M : T}
      \end{aligned}$
    \item Let $ \Gamma = {\Gamma_1} , {\Gamma_2}$ with ${\Gamma_1 \vdash \lambda \abstr{x}M : S \multimap T}$ and
      ${\Gamma_2 \vdash N : S}$. Let $f,g$ denote the morphisms $\lrb{{\Gamma_1, x : S \vdash M : T}}$, $\lrb{{\Gamma_2 \vdash N : S}}$ respectively.
      \begin{align*}
	\lrb{{\Gamma \vdash (\lambda \abstr{x}M)~N : T}} 
	&=\textit{eval}_{ T,  S} \circ (\Phi(f) \otimes g) \\
	&=\varepsilon_{\lrb{ S}} \circ (\Phi(f) \otimes id_{\lrb{ S}}) \circ (id_{\lrb{ \Gamma_1}} \otimes g) \\
	&=\varepsilon_{\lrb{ S}} \circ \_ \otimes \lrb{ S}(\Phi (f)) \circ (id_{\lrb{{\Gamma_1}}} \otimes g) \\
	&=f \circ (id_{\lrb{{\Gamma_1}}} \otimes g) \\
	&=\lrb{{\Gamma \vdash (N/x)M : T}}
      \end{align*}
    \item Let $\p \Gamma = \p{\Gamma_1}, \p{\Gamma_2}$ with $\p{\Gamma_1 \vdash \pair{m_1}{m_2} : Q \with R}$ and $\p{\Gamma_2 , x : Q \vdash n : P}$.
      \begin{align*}
	&\lrb{\p{\Gamma \vdash \elimwith^i(\pair{m_i}{m_2}, \abstr{x}o) : P}}\\
	&=\lrb{\p{\Gamma_2, x : Q \vdash n : P}} \circ (\pi_i \otimes id)
	\circ (\lrb{\p{\Gamma_1 \vdash \pair{m_1}{m_2} : Q \with R}} \otimes id)\\
	&=\lrb{\p{\Gamma_2 , x : Q \vdash n : P}} \circ (\lrb{\p{\Gamma_1 \vdash m_i : Q}} \otimes id) \circ
	\sigma_{\lrb{\p{\Gamma_1}}, \lrb{\p{\Gamma_2}}}\\
	&=\lrb{\p{\Gamma \vdash (m_i/x)o : P}}
      \end{align*}
    \item Let $\Gamma = \Gamma_1, \Gamma_2$, with $\Gamma_1 \vdash \inl(m) : B \oplus C$,
      $\Gamma_2 , x : B \vdash n : A$ and $\Gamma_2, y : C \vdash o : A$. Let $f,g$ denote the morphisms
      $\lrb{\Gamma_2 , x : B \vdash n : A}$, $\lrb{\Gamma_2, y : C \vdash o : A}$ respectively.
      \begin{align*}
	\lrb{\Gamma_1, \Gamma_2 \vdash \elimplus(\inl(m), \abstr{x}n , \abstr{y}o) : C}
	&= [f,g] \circ d \circ \sigma \circ (\lrb{\inl(m)}\otimes id)\\
	&= [f,g] \circ d \circ \sigma \circ (i_1 \otimes id) \circ (\lrb{m} \otimes id)\\
	&= [f,g] \circ d \circ (id \otimes i_1) \circ \sigma \circ (\lrb{m} \otimes id)\\
	&= f \circ \sigma \circ (\lrb{m} \otimes id)\\
	&= \lrb{\Gamma_1, \Gamma_2 \vdash (m/x) n : A}
      \end{align*}
    \item This case is analogous to the previous one.
    \item Let $\Gamma = \Gamma_1, \Delta_1, \Delta_2$ with
      $\Gamma_1 , x : S_1, x : S_2 \vdash N : T$, $\Delta_1 \vdash M_1 : S_1$ and $\Delta_2 \vdash M_2 : S_2$.
      \begin{align*}
	      &\lrb{\Gamma_1, \Delta_1, \Delta_2 \vdash \elimtens(M_1 \otimes M_2, \abstr{xy}N) : T}\\
	      &=\lrb{\Gamma_1 , x : S_1, x : S_2 \vdash N : T} \circ (id \otimes \lrb{M_1} \otimes \lrb{M_2})\\
	      & = \lrb{\Gamma \vdash (M_1/x, M_2/y)N : T}
	      \tag*{\qed}
      \end{align*}
  \end{enumerate}
\end{proof}

\soundnesscor*
\begin{proof}
  By induction on the reduction relation.
  First, notice that if $M \leftrightarrows N$, then $\lrb{M} = \lrb{N}$, by Lemma~\ref{lem:pm_eq}.
  Thus, we only need to show the theorem for the case $M\rightarrow N$ (see Figure~\ref{fig:red_rightarrow}).
  \begin{itemize}
    \item Rules~\eqref{rule:one}
      to~\eqref{rule:tensor}, are covered by
      Lemma~\ref{lem:subst_eq_p}
    \item Rules~\eqref{rule:commtop} to~\eqref{rule:comminr}
      are direct consequence of Lemma~\ref{lem:pm_eq}
\item Rule \eqref{rule:casting}:
	We have $\lrb{\cdot \vdash \tau(\mathcal B(\p{v \otimes w}))} = \tau \circ \lrb{\cdot \vdash \mathcal B(\p{v \otimes w})} =
		  \lrb{\cdot \vdash \mathcal B(\p v)} \otimes \lrb{\cdot \vdash \mathcal B(w)}
		  =\lrb{\cdot\vdash\mathcal B(\p v)\otimes\mathcal B(w)}$.
	  \item Rule \eqref{rule:commtau}:
		We have $\lrb{\Gamma \vdash \tau (\textstyle\sum_{i=1}^n p_i \sdot v_{bi})} = \tau \circ \textstyle\sum_{i=1}^n p_i \sdot \lrb{v_{bi}} = \textstyle\sum_{i=1}^n p_i \tau \circ \lrb{v_{bi}}
		=\lrb{\textstyle\sum_{i=1}^n p_i\sdot \tau (v_{bi})}$,
		because $\tau$ is a completely positive isomorphism.
    \item Rule~\eqref{rule:commB}:
      $\begin{aligned}[t]
	\lrb{\vdash \mathcal B(\p{a}\sdot \p{v_b}) : \mathcal B(\p P)}
	&= \mathcal B(\p a \sdot \lrb{\p{v_b}})\\
	&= \p a \sdot \lrb{\p{v_b}} \circ \_ \circ (\p a \sdot \lrb{\p{v_b}})^*\\
	&= \p a \sdot \lrb{\p{v_b}} \circ \_ \circ \p{a^*} \sdot \lrb{\p{v_b}}^*\\
	&= \p{aa^*} \sdot \lrb{\p{v_b}} \circ \_ \circ \lrb{\p{v_b}}^*\\
	&= |\p a|^2 \sdot \mathcal B(\p{v_b})\\
	&=\lrb{\vdash |\p a|^2 \sdot \mathcal B(\p{v_b}) : \mathcal B(\p P)}
      \end{aligned}$

    \item Rule~\eqref{rule:commBapp}:
      Since the rule $\mathcal B^m$ applies only to closed terms, the terms
      in this rule are necessarily closed.  Moreover, by
      Theorem~\ref{thm:prog}, the first values
      must be lambda terms.  To simplify the proof, we first consider the case
      $n = k = 1$ and $p_1 = q_1 = 1$, and then generalise the result.
      \begin{itemize}
	\item 
	  Let $f = \lrb{\p{\vdash \lambda \abstr{x}m : P \multimap Q}}$, and $g = \lrb{\p{\vdash v : P}}$.
	  We have,
	  \begin{align*}
	    \lrb{\vdash \mathcal B(\p{\lambda \abstr{x}m})~\mathcal B(\p v) : \mathcal B(\p Q)}
	    &= \mathcal B (eval_{\lrb{\p P}, \lrb{\p Q}}) \circ \tau\circ (\mathcal B(f) \otimes\mathcal B(g))\\
	    &= \mathcal B (eval_{\lrb{\p P}, \lrb{\p Q}} \circ (f \otimes g))\\
	    &= \mathcal B (\lrb{\p{\vdash (\lambda \abstr{x}m)~v : Q)}}
	  \end{align*}
	\item In the general case, we have $M=(\textstyle\sum_{i=1}^n p_i \sdot
	  \mathcal B(\p{m_i}))~(\textstyle\sum_{j=1}^k q_j \sdot \mathcal
	  B(\p{m'_j}))$ and $N= \textstyle\sum_{i,j} \mathcal B(\p{m_i~m'_j})$.
	  The equality of the denotations of these two terms follows from the
	  case above and Lemma~\ref{lem:pm_eq}.
      \end{itemize}
    \item The contextual rules are straightforward and we omit them.
      \qed
  \end{itemize}
\end{proof}

\adequacy*
\begin{proof} 
  Let $[.] \vdash E : T$, with $T \in \{\one, 1, \mathcal B(\one)\}$. Case analysis.
  \begin{itemize}
    \item Let $T = \one$ or $1$:
      By Theorem~\ref{thm:prog}, we have $E[M] \hookrightarrow^* \alpha\sdot \ast$ and $E[N] \hookrightarrow^* \beta\sdot \ast$. 
      Then, by Theorem~\ref{thm:soundnesshook}
      and Lemma~\ref{lem:subst_d},
      \[
	\lrb{\vdash\alpha\sdot\ast:1}
	= \lrb{\vdash E[M]:1}
	= \lrb{\vdash E[N]:1}
	= \lrb{\vdash\beta\sdot\ast:1}
      \]
      Thus, $\alpha = \beta$.

    \item  Let $T = \mathcal B(\one)$:
      By Theorem~\ref{thm:prog}, we have $E[M] \hookrightarrow^* p\sdot\mathcal B(\p\ast)$ and $E[N] \hookrightarrow^* q\sdot \mathcal B(\p\ast)$. 
      Then, by Theorem~\ref{thm:soundnesshook} and Lemma~\ref{lem:subst_d},
      \begin{align*}
	\lrb{\vdash p\sdot\mathcal B(\p\ast):\mathcal B(\one)}
	&= \lrb{\vdash E[M]:\mathcal B(\one)}\\
	&= \lrb{\vdash E[N]:\mathcal B(\one)}\\
	&= \lrb{\vdash q\sdot\mathcal B(\p\ast):\mathcal B(\one)}
      \end{align*}
      Thus, $p = q$.
  \end{itemize}
  Therefore, $M\sim N$.
  \qed
\end{proof}
\end{document}